\documentclass[journal]{IEEEtran}

\usepackage[utf8]{inputenc} 
\usepackage[T1]{fontenc}    
\usepackage{graphicx}
\usepackage{bm}
\usepackage[font=footnotesize]{caption}
\usepackage{hyperref}
\usepackage[noend]{algorithm2e}
\usepackage{nicefrac}       
\usepackage[binary-units=true]{siunitx}
\usepackage{amsfonts}       
\usepackage{url}

\usepackage{subcaption}

\usepackage{booktabs}

\usepackage{localmath}

\usepackage{pifont}
\providecommand{\vvt}{\tilde{\vv}}

\providecommand{\vqt}{\tilde{\vq}}
\providecommand{\half}{\nicefrac{1}{2}}

\graphicspath{{figures/}{figures/experiment_iva/}{figures/experiment_head/}{figures/speed_contest_journal_ipa_20201222-185917_2bc728a163/}}

\usepackage{textcomp}
\usepackage[absolute,showboxes]{textpos}

\TPMargin{5pt}

\newcommand{\copyrightstatement}{
	\begin{textblock}{0.84}(0.08,0.01)    
		\noindent
		\footnotesize
		\textcopyright~2021~IEEE.  Personal  use  of  this  material  is  permitted.
		Permission from IEEE must be obtained for all other uses, in any current or
		future media, including reprinting/republishing this material for advertising
		or promotional purposes, creating new collective works, for resale or
		redistribution to servers or lists, or reuse of any copyrighted component of
		this work in other works.
	\end{textblock}
}

\begin{document}

\copyrightstatement

\title{Independent Vector Analysis via Log-Quadratically Penalized Quadratic Minimization}

\author{Robin~Scheibler,~\IEEEmembership{Senior Member,~IEEE}
\thanks{LINE Corporation, Tokyo, 160-0004, Japan (e-mail: robin.scheibler@linecorp.com)}
\thanks{The software to reproduce the results of this paper is available at \protect\url{https://github.com/fakufaku/auxiva-ipa}.}
}


\maketitle

\thispagestyle{empty}

\begin{abstract}
  We propose a new algorithm for blind source separation (BSS) using independent vector analysis (IVA).
This is an improvement over the popular auxiliary function based IVA (AuxIVA) with iterative projection (IP) or iterative source steering (ISS).
We introduce \textit{iterative projection with adjustment} (IPA), where we update one demixing filter and \textit{jointly} adjust all the other sources along its current direction.
Each update involves solving a non-convex minimization problem that we term \textit{log-quadratically penalized quadratic minimization} (LQPQM), that we think is of interest beyond this work.
In the general case, we show that its global minimum corresponds to the largest root of a univariate function, reminiscent of modified eigenvalue problems.
We propose a simple procedure based on Newton-Raphson to efficiently compute it.
Numerical experiments demonstrate the effectiveness of the proposed method.
First, we show that it efficiently decreases the value of the surrogate function.
In further experiments on synthetic mixtures, we study the probability of finding the true demixing matrix and convergence speed.
We show that the proposed method combines high success rate and fast convergence.
Finally, we validate the performance on a reverberant blind speech separation task.
We find that all the AuxIVA-based methods perform similarly in terms of acoustic BSS metrics.
However, AuxIVA-IPA converges faster.
We measure up to 8.5 times speed-up in terms of runtime compared to the next best AuxIVA-based method, depending on the number of channels and the signal-to-noise ratio (SNR).


\end{abstract}


%
\IEEEpeerreviewmaketitle

\section{Introduction}

\IEEEPARstart{B}{lind} source separation (BSS) deals with decomposing a mixture of signals into its constitutive components with as little prior information as possible~\cite{Comon:1512057}.
It has found prominent application in multichannel audio processing~\cite{Makino:2018iq}, e.g., for the separation of speech~\cite{Makino:2007vg} and  music~\cite{Cano:2019dw}, but also in biomedical signal processing for electrocardiogram~\cite{ZARZOSO:1997dt} and electroencephalogram~\cite{Cong2019}, and in digital communications~\cite{Yang:2018kr}.
For multichannel signals, independent component analysis (ICA) allows to do BSS, only requiring statistical independence of the sources and some other mild conditions~\cite{Comon:1994kr}.
Independent Vector Analysis (IVA) is an extension of ICA to the analysis of multiple parallel mixtures where sources within one mixture are mutually independent, but may be dependent on at most one source in each of the other mixtures~\cite{Hiroe:2006ib,Kim:2006ex,Lee:2007bw}.
Such problems appear, for example, in convolutive source separation in the frequency domain~\cite{Smaragdis:1998kl}, or in the analysis of fMRI data~\cite{Lee:2008dc}.
Without further considerations, separating each mixture individually with ICA introduces a permutation ambiguity where the order of extracted sources may be different for each of them.
IVA solves this problem by assuming a multivariate distribution of the sources over the multiple mixtures and doing the separation jointly.
The source model is used to express the likelihood of the input data which is then maximized to estimate the source signals.
This optimization problem is non-convex, and, without a known closed form solution.
Auxiliary function based IVA (AuxIVA) was proposed as a fast and stable optimization method to solve IVA~\cite{Ono:2011tn}.
It relies on the majorization-minimization (MM) technique~\cite{Lange:2016wp} and is applicable to super-Gaussian source models.
AuxIVA majorizes the IVA cost function with a quadratic surrogate, leading to an iterative algorithm were a so-called ``sequentially drilled'' joint congruence\footnote{Also known as hybrid exact-approximate diagonalization (HEAD)~\cite{Yeredor:hr}.} (SeDJoCo) problem~\cite{Yeredor:hr,Yeredor:2010kl,Weiss:2017il} must be solved at every iteration.
Solving the SeDJoCo in closed-form for more than two sources is still an open problem and instead AuxIVA performs alternating minimization of the surrogate with respect to the demixing filters of the sources~\cite{Ono:2011tn}.
This approach has been coined \textit{iterative projection} (IP).
A similar solution was also proposed in the context of semi-blind Gaussian source separation~\cite{Degerine:2004dy}.
Alternatives to the MM approach have been proposed.
Originally, the natural gradient (NG) algorithm~\cite{Hiroe:2006ib,Kim:2006ex}, as well as FastIVA~\cite{Lee:2007ct}, a fixed-point algorithm, were proposed to carry out the minimization.
Proximal splitting allows for a versatile algorithm with a heuristic extension based on masking~\cite{Yatabe:2018vb, Yatabe:2020tv}.
Another approach, specialized for two sources, is based on expectation-maximization and a Gaussian mixture model~\cite{Gu:2019hl}.

This paper focuses on the MM approach which underpins many algorithms with more sophisticated source models.
These include non-negative low-rank~\cite{Kitamura:2016vj}, based on a variational auto-encoder~\cite{Kameoka:2019be}, a deep network~\cite{Makishima:2019fl}, or using inter-clique dependence~\cite{Shin:2020fg}.
In addition, it can also be applied to overdetermined IVA (OverIVA), i.e., when there are more channels than sources~\cite{Scheibler:2019vx}.
As such, any improvement to the optimization of the surrogate function in AuxIVA directly translates to improvements for all of these algorithms.
For two sources, the SeDJoCo problem can be solved by a generalized eigenvalue decomposition~\cite{Ono:2010hh,Ono:2012wa} and thus globally optimal updates of the surrogate are possible.
A similar situation arises for blind extraction of a single source with the fast independent vector extraction algorithm~\cite{Scheibler:2019tt,Ikeshita:2020ic}.
For three and more sources, iterative projection 2 (IP2) does pairwise updates of two sources at a time, leading to faster convergence~\cite{Ono:2018wk,Scheibler:2020mm}.
Finally, \textit{iterative source steering (ISS)} performs a series of rank-1 updates of the demixing matrix which correspond in fact to alternating updates of the columns of the mixing matrix~\cite{Scheibler:2020ig}.
While the convergence of ISS is similar to that of IP, it does not require matrix inversion, and has an overall lower computational complexity.
Thus, when separating three and more sources,  all of IP, IP2, and ISS, fix all the other sources when doing one of the updates.
This means that further correction can only happen at the next iteration.

In this work, we propose \textit{iterative projection with adjustment} (IPA), a joint update of one demixing filter with an extra rank-1 modification of the rest of the demixing matrix.
As opposed to IP, IP2, and ISS, when updating the demixing filter of one source, we simultaneously correct the demixing filters of all other sources accordingly.
Intuitively, this allows the algorithm to make progress in the demixing of \textit{all} sources at \textit{every} update.
Concretely, we adopt a multiplicative update form where the current demixing matrix is multiplied by a rank-2 perturbation of the identity matrix.
We show that the minimization of the IVA surrogate function with respect to the multiplicative update leads to an optimization problem that we believe is of independent interest.
We term this problem \textit{log-quadratically penalized quadratic minimization} (LQPQM).
\begin{problem}[LQPQM]
  \label{problem:lqpqm}
  Let $\mA,\mC\in \C^{d\times d}$ be Hermitian positive definite and semi-definite, respectively, and $\vb,\vd\in\C^d$, and $z\in \R$, $z\geq 0$.
  Then, the LQPQM problem is,
  \begin{equation}
    \underset{\vq\in\C^{d}}{\min}\ (\vq - \vb)^\H \mA(\vq - \vb) - \log \left((\vq - \vd)^\H \mC (\vq - \vd) + z\right).
    \tag{P1}\elabel{lqpqm}
  \end{equation}
\end{problem}
For a sneak peek of what the objective function looks like in two dimensions, skip to \ffref{lqpqm_landscape}.
One of the main contributions of this paper is to show that, despite being non-convex, the global minimum of \eref{lqpqm} can be computed efficiently.
In the general case, we show that all the stationary points of the objective of \eref{lqpqm} can be characterized as the zeros of a univariate non-linear equation.
Then, we prove that the value of the objective function decreases for increasing values of the zeros, and the global minimum thus corresponds to the largest zero.
Furthermore, we find that its location is the only zero larger than the largest generalized eigenvalue for the problem $\mC \vq = \varphi \mA \vq$, $\varphi\in\R$.
Thus, we propose to use the Newton-Raphson root finding algorithm in this interval.
The procedure we propose is reminiscent of other algorithms for problems involving pairs of quadratic forms such as modified eigenvalue problems~\cite{Golub:1973ha,Bunch:1978bn,Yu:1991cy}, generalized trust region subproblems~\cite{MORE:1993dy}, or some applications in robust beamforming~\cite{Lorenz:em}, multi-lateration~\cite{Beck:2012dj}, or direction of arrival estimation~\cite{Togami:2020doa}.

We validate the performance of the proposed method via comprehensive numerical experiments.
First, we evaluate the effectiveness of IPA to solve the SeDJoCo sub-problem compared to existing algorithms.
IPA is found to be the most effective to reduce the value of the surrogate function in a single iteration.
Coupled with the guaranteed monotonical decrease of the surrogate, this demonstrates its potential for use within AuxIVA.
The second experiment compares the different flavors of AuxIVA, for what we believe is the first time, in terms of convergence speed and probability of success to recover the true demixing matrix.
We find that the faster update rules, such as IP2 and the proposed IPA, are also more likely to find the correct solution.
FastIVA is found to require the least iterations on synthetic data, but finds spurious solutions more often.
IPA is the second fastest.
Our last experiment is on the downstream task of separating multichannel speech mixtures in the frequency domain.
We find all AuxIVA-based methods to perform equally well in terms of common audio BSS metrics.
However, the proposed method outperforms them in speed of convergence.
FastIVA is found to be very competitive, with slightly faster convergence for 5 and 6 channels at low SNR.
We measure up to $8.5\times$ and $2.4\times$ speed-up compared to other AuxIVA-based methods and FastIVA, respectively.

The rest of this paper is organized as follows.
We cover the background on IVA, MM optimization, and AuxIVA in~\sref{background}.
\sref{ipa} describes IPA, the proposed AuxIVA updates, and proves that they are given by the solution to an LQPQM.
The procedure to find the global minimum of an LQPQM is stated and proved in \sref{LQPQM}.
We evaluate the performance of AuxIVA with IPA updates and compare to IP, ISS, and IP2, as well as the natural gradient~\cite{Hiroe:2006ib,Kim:2006ex} and FastIVA~\cite{Lee:2007ct} in~\sref{experiments}.
\sref{conclusion} concludes this paper.

\section{Background}
\seclabel{background}

We consider the determined separation problem with $F$ mixtures of $M$ sources, recorded by $M$ sensors,
\begin{equation}
  \vx_{fn} = \mA_{f} \vs_{fn}, \quad n = 1,\ldots, N,
  \elabel{mixture_model}
\end{equation}
where $\vx_{fn} \in \C^M$ and $\vs_{fn}$ are the measurement and source vectors, respectively, in mixture $f$ and at time $n$.
Here, $\mA_f\in\C^{M\times M}$ is the mixing matrix whose entry $(\mA_f)_{mk}$ is the transfer function from source $k$ to sensor $m$.
Such parallel mixtures most frequently appear as the result of time-frequency domain processing for the separation of convolutional mixtures, e.g., of audio sources~\cite{Smaragdis:1998kl}.
\sref{convolutive_mixtures} briefly explains how the complex mixture model~\eref{mixture_model} is obtained from the real-valued signals recorded by microphones.
In this case, the separation may be done by finding the $M\times M$ demixing matrices,
\begin{equation}
  \mW_{f} = \begin{bmatrix} \vw_{1f} & \cdots & \vw_{Mf} \end{bmatrix}^\H,\quad f=1,\ldots,F,
\end{equation}
such that an estimate of the sources is,
\begin{equation}
  \hat{\vs}_{fn} = \mW_f \vx_{fn}.
\end{equation}
Thus, row $k$ of $\mW_f$ contains the \textit{demixing filter} $\vw_{kf}^\H$ for source $k$, and $\hat{\vs}_{fn}$ is the estimated source vector.
Estimating matrices
\begin{equation}
  \calW = \{\mW_f\,:\, f=1,\ldots,F\}
\end{equation}
from the observed vectors $\vx_{fn}$ is the purpose of IVA.

In the rest of the manuscript, we use lower and upper case bold letters for vectors and matrices, respectively.
Furthermore, $\mA^\top$, $\mA^\H$, and $\det(\mA)$ denote the transpose, conjugate transpose, and determinant of matrix $\mA$, respectively.
The diagonal matrix with entries $a_1,\ldots,a_d$ is denoted $\diag(a_1, \ldots, a_d)$.
A bold zero, i.e., $\vzero$, is the all zero vector or matrix of the context-appropriate shape.
Optimizers of optimization problems are denoted by a star, e.g., $\vx^\star$.
This is not to be confused with complex conjugation denoted by an asterisk, e.g., $z^*$ is the complex conjugate of scalar $z\in \C$.
Let $\vv \in \C^d$, a complex $d$-dimensional vector.
The vector $\vv^*$ contains the conjugated coefficients of $\vv$.
The Euclidean norm of $\vv$ is $\| \vv \| = (\vv^\H \vv)^{\half}$.
Unless specified otherwise, indices $f$, $k$, $m$, and $n$ always take the ranges defined in this section, i.e., from 1 to $F$, $M$, $M$, and $N$, respectively.
We omit the bounds of sums and products over these indices when they span the ranges just defined.
For example, $\sum_k$ is from $k=1$ to $M$, and $\sum_{kn}$ is a double sum over $k=1$ to $M$ and $n=1$ to $N$.

\subsection{Independent Vector Analysis}

IVA can be specified either as minimization of the Kullback-Leibler divergence~\cite{Hiroe:2006ib,Kim:2006ex}, or as a maximum likelihood estimation problem~\cite{Anderson:2014jd}.
Here, we follow the latter approach.
The observed data are the mixture vectors $\vx_{fn}$, and the parameters to estimate are the demixing matrices $\mW_f$.
We define the $k$th source component vector (SCV), at time $n$, as
\begin{align}
  \check{\vs}_{kn} = \begin{bmatrix} s_{k1n} & \cdots & s_{kFn} \end{bmatrix}^\top.
  \elabel{src_freq_vec}
\end{align}
The likelihood function is derived on the basis of the two following assumptions.
\begin{assumption}[Independence of Sources]
  The sources are statistically independent, i.e., their joint distribution is the product of the marginals.
\end{assumption}
\begin{assumption}[Source Model]
  The sources follow a multivariate distribution, i.e.,
  \begin{equation}
    p_S(\check{\vs}_{kn}) = \frac{1}{Z} e^{-F(\check{\vs}_{kn})},\quad \forall k
  \end{equation}
  where $F(\vs)$ is called the contrast function and $Z$ is a normalizing constant that does not depend on the source.
\end{assumption}
Let us denote the estimated source $\hat{s}_{kfn} = \vw_f^\H \vx_{fn}$ and define $\bar{\vs}_{kn}$ similarly to \eref{src_freq_vec},
\begin{equation}
  \bar{\vs}_{kn} = \begin{bmatrix} \hat{s}_{k1n} & \cdots & \hat{s}_{kFn} \end{bmatrix}^\top.
  \elabel{est_source_vec}
\end{equation}
By further using independence, the joint distribution of the sources is just the product of their marginals.
Thus, the likelihood of the observation is
\begin{align}
  \calL(\calW) & = \prod_n p_X(\vx_{1n},\ldots,\vx_{Fn}) \\
               & = \prod_{kn} p_S(\bar{\vs}_{kn}) \prod_{f} | \det\mW_f |^{2N},
\end{align}
where $p_X$ is the probability density function of the observed signals.
The determinant term is due to the change of variable necessary to introduce $\bar{\vs}_{kn}$.
Then, $\mW_f$ is estimated by minimizing the negative log-likelihood function, shown here with constant terms omitted,
\begin{equation}
  \ell(\calW) = \sum\nolimits_{kn}F(\bar{\vs}_{kn}) - 2 N \sum\nolimits_f \log |\det \mW_f|.
  \elabel{cost_iva}
\end{equation}
The choice of the contrast function and the minimization of the negative log-likelihood have been the object of considerable work~\cite{Hiroe:2006ib,Kim:2006ex,Lee:2007bw,Ono:2010hh,Ono:2011tn,Ono:2018wk,Scheibler:2020mm}.
Source models based on spherical super-Gaussian distributions~\cite{Ono:2010hh,Ono:2011tn,Scheibler:2020mm} underpin AuxIVA, described in~\sref{auxiva}.
They are motivated by the sparsity of signal power over time in many applications, including speech.
Conveniently, they allow to apply the MM optimization technique that we describe next.
The function~\eref{cost_iva} is non-convex, and thus its optimization focuses on finding a local minimum.

\subsection{Majorization-Minimization Optimization}
\seclabel{mm_opt}

MM optimization is an iterative technique that makes use of a \textit{surrogate} function that is both tangent to, and majorizes the cost function everywhere.
Under appropriate regularity conditions, repeatedly minimizing the surrogate leads to a stationary point of the original cost function~\cite{Lange:2016wp}, usually a local minimum, but counterexamples exist~\cite{Wu:1983hp}.
\begin{proposition}[MM Monotonic Descent~\cite{Lange:2016wp}]
  \label{prop:mm}
  Let $f(\theta)$ and $Q(\vtheta, \hat{\vtheta})$ be functions such that
  \begin{align}
    Q(\hat{\vtheta}, \hat{\vtheta}) = f(\hat{\vtheta}),\quad \text{and},\quad  Q(\vtheta, \hat{\vtheta}) \geq f(\vtheta), \quad \forall \vtheta,\hat{\vtheta}. \elabel{surrogate_eq_ineq}
  \end{align}
  The function $Q$ is a surrogate function for the cost function $f$.
  Given an initial point $\vtheta_0$, consider the sequence,
  \begin{align}
    \vtheta_{t} = \underset{\vtheta}{\arg\min}\ Q(\vtheta, \vtheta_{t-1}), \quad t=1,\ldots, T.
    \elabel{mm_min_surrogate}
  \end{align}
  Then, the cost function is monotonically decreasing on the sequence, $\vtheta_0, \vtheta_1, ..., \vtheta_T$, i.e.,
  \begin{align}
    f(\vtheta_0) \geq f(\vtheta_1) \geq \ldots \geq f(\vtheta_T).
    \elabel{mm_update}
  \end{align}
\end{proposition}
The proof of this proposition is easily established from~\eref{surrogate_eq_ineq} and~\eref{mm_min_surrogate}.
Note that the proposition still holds even if the minimization in \eref{mm_min_surrogate} is replaced by any operation that merely reduces the value of $Q(\vtheta, \vtheta_{t-1})$, i.e., such that
\begin{align}
  Q(\vtheta_{t},\vtheta_{t-1}) \leq Q(\vtheta_{t-1}, \vtheta_{t-1}).
  \elabel{descent_property}
\end{align}
MM optimization has many desirable properties.
It allows to tackle non-convex and/or non-smooth objectives.
Unlike gradient descent, it does not require tuning of a step size.
Finally, the derived updates often have an intuitive interpretation.
It has been applied to multi-dimensional scaling~\cite{DeLeeuw:1997vw}, sparse norm minimization as the popular iteratively reweighted least-squares algorithm~\cite{Daubechies:2010hf}, sub-sample time delay estimation~\cite{Yamaoka:2019wb}, and direction-of-arrival estimation~\cite{Togami:2020doa}.
For in-depth theory, a general introduction, or more applications in signal processing, see \cite{Lange:2016wp,Hunter:2004hq,Sun:2017hb}.

\begin{algorithm}[t]
  \SetKwInOut{Input}{Input}\SetKwInOut{Output}{Output}
  \SetKw{KwBy}{by}
  \Input{Mixture signals $\vx_{fn}\in\C^M$, $\forall f,n$}
  \Output{Separated signals $\hat{\vs}_{fn}\in\C^M$, $\forall f,n$}
  \DontPrintSemicolon
  $\mW_f \gets \mI_K,\ \forall f$\;
  $\hat{\vs}_{fn} \gets \vx_{fn},\ \forall f,n$\;
  \For{loop $\leftarrow 1$ \KwTo $\text{max. iterations}$}{
    $r_{kn} \gets \sqrt{\sum_f |\hat{s}_{kfn}|^2},\ \forall k,n$\;
    $\mV_{kf} \gets \frac{1}{N} \sum_n \frac{G^\prime(r_{kn})}{2 r_{kn}} \vx_{fn} \vx_{fn}^\H, \ \forall k,f$\;
    \For{$f \gets 1$ \KwTo $F$}{
      $\mW_f \gets \operatorname{Update}(\mW_f, \mV_{1f},\ldots,\mV_{Mf})$\;
      $\hat{\vs}_{fn} \gets \mW_f \vx_{fn},\ \forall n$\;
    }
  }
  \vspace{0.5cm}
  \caption{AuxIVA. The sub-routine \textit{Update} performs one of IP, IP2, ISS, or IPA.}
  \label{alg:auxiva}
\end{algorithm}

\subsection{Auxiliary function based IVA}
\seclabel{auxiva}

AuxIVA applies the MM technique to the IVA cost function \eref{cost_iva}~\cite{Ono:2011tn}.
This is done by restricting the contrast function to the class of spherical super-Gaussian source models.
\begin{definition}[Spherical super-Gaussian contrast function~\cite{Ono:2010hh}]
  \label{def:supergauss}
  A spherical super-Gaussian contrast function depends only on the magnitude of the SCV, i.e.,
  \begin{equation}
    F(\check{\vs}_{kn}) = G(\|\check{\vs}_{kn}\|)
    \elabel{spherical}
  \end{equation}
  and, in addition, $G\,:\,\R_+\to\R$ is a real, continuous, and differentiable function such that $G^\prime(r)/r$ is continuous everywhere and monotonically decreasing for $r>0$.
  The function $G^\prime(r)$ is the derivative of $G(r)$.
\end{definition}
These contrast functions include Laplace, time-varying Gauss, Cauchy, and other popular source models~\cite{Ono:2010hh,Scheibler:2020mm}.
They can also be majorized by a quadratic function.
\begin{lemma}[Theorem~1 in~\cite{Ono:2010hh}]
  Let $G$ be as in Definition~\ref{def:supergauss}. Then,
  \begin{equation}
    G(r) \leq G^\prime(r_0)\frac{r^2}{2 r_0} + \left( G(r_0) - \frac{r_0}{2} G^\prime(r_0)\right),
  \end{equation}
  with equality for $r=r_0$.
\end{lemma}
Equipped with this inequality, we can form $\ell_2$, a surrogate of \eref{cost_iva} such that $\ell(\calW) \leq N \ell_2(\calW) + \text{constant}$,
\begin{multline}
  \ell_2(\calW) = \sum\nolimits_{kf} \vw_{kf}^\H \mV_{kf} \vw_{kf} - 2 \sum\nolimits_f \log |\det \mW_f |,
  \elabel{cost_auxiva}
\end{multline}
where
\begin{equation}
  \mV_{kf} = \frac{1}{N} \sum\nolimits_n \frac{G^\prime(r_{kn})}{2r_{kn}} \vx_{fn} \vx_{fn}^\H,
  \elabel{auxiliary_variable_update}
\end{equation}
and $r_{kn}$ is an auxiliary variable.
The resulting MM optimization algorithm is described in \algref{auxiva} where \texttt{Update} is a sub-routine that decreases the value of the surrogate $\ell_2(\calW)$.
Conveniently, the surrogate is separable for $f$.
While this may seem counter-intuitive, information is shared between mixtures by the update of $r_{kn}$ at every iteration.
Taking $r_{kn} = \| \bar{\vs}_{kn} \|$, with $\bar{\vs}_{kn}$ from \eref{est_source_vec}, ensures that the surrogate is tangent to the objective, i.e. \eref{surrogate_eq_ineq} (left).
Then, Proposition~\ref{prop:mm} guarantees monotonic decrease of the original objective, \eref{cost_iva}, where the different mixtures are dependent.
Interestingly, $r_{kn}$ is the magnitude  of the source estimate from the previous iteration.
Several choices are already available for the \texttt{Update} sub-routine in \algref{auxiva}.
Closed-form minimization of \eref{cost_auxiva} is possible for two sources, and the resulting AuxIVA algorithm is very fast~\cite{Ono:2012wa}.
However, for more than two sources, it is still an open problem.
Instead, a number of strategies updating the parameters alternatingly in a block-coordinate descent fashion have been proposed.
A visual summary of these different strategies is shown in \tref{auxiva_visual_updates}.

One of them is IP~\cite{Ono:2011tn,Degerine:2004dy}.
It considers minimization of \eref{cost_auxiva} with respect to only one demixing filter, e.g., $\vw_{kf}$, keeping everything else fixed, with closed-form solution,
\begin{align}
  \vw_{kf} \gets \frac{(\mW_f\mV_{kf})^{-1} \ve_k}{\sqrt{\ve_k^\top \mW_f^{-\H} \mV_{kf}^{-1}\mW_f^{-1} \ve_k}}.
\end{align}
The update is applied for $k=1,\ldots,M$, in order.

IP2 is an improvement over IP in which \eref{cost_auxiva} is minimized with respect to two demixing filters, e.g. $\vw_{kf}, \vw_{mf}$, keeping everything else fixed~\cite{Ono:2018wk,Scheibler:2020mm}.
First, form $\mP_{uf} = (\mW_f \mV_{uf})^{-1}[\ve_k\, \ve_m]$, and let $\wt{\mV}_{uf} = \mP_{uf}^\H \mV_{uf} \mP_{uf}$, for $u=k,m$.
Then, the new demixing filters are given by the generalized eigenvectors $\vw$ of the generalized eigenvalue problem $\wt{\mV}_{kf} \vw = \varphi \wt{\mV}_{mf} \vw$, with $\varphi \in \R$.
The update is applied with $k=(2k^\prime \mod M)$, $m=(2k^\prime+1 \mod M)$ for $k^\prime=1,\ldots,M$.

Finally, ISS updates the whole demixing matrix~\cite{Scheibler:2020ig},
\begin{equation}
  \mW_f \gets \mW_f - \vv_{kf} \vw_{kf}^\H,
  \elabel{iss_update}
\end{equation}
where the $m$th coefficient of $\vv_{kf}$ is given by
\begin{equation}
  v_{mkf} = \begin{cases}
    \frac{\vw_{mf}^\H \mV_{mf} \vw_{kf}}{\vw_{kf}^\H \mV_{mf} \vw_{kf}} & \text{if $m \neq k$,} \\
    1 - (\vw_{kf}^\H \mV_{kf} \vw_{kf})^{-\half} & \text{if $m=k$}.
  \end{cases}
  \elabel{iss_update_2}
\end{equation}
One can show that~\eref{iss_update_2} corresponds to an update of the $k$th column of the mixing matrix, i.e., $\mW_f^{-1}$~\cite{Scheibler:2020ig}.
This is performed for $k=1,\ldots,M$, in order, once per iteration.

Finally, a Newton-Conjugate Gradient (NCG) scheme has been proposed to solve SeDJoCo problems~\cite{Yeredor:hr,Yeredor:2012gv}, of which~\eref{cost_auxiva} is an instance.
Newton method is very fast and converges quadratically when initialized in the vicinity of a stationary point.
However, it does not distinguish between minima and maxima, and might increase the value of~\eref{cost_auxiva}.
It is thus not directly applicable to the construction of an MM algorithm as it does not ensure the descent property~\eref{descent_property}.

\subsection{Interpretation of AuxIVA as Iterative Gaussian Separation}

SeDJoCo problems have been introduced in the context of the semi-blind separation of Gaussian sources~\cite{Yeredor:2010kl,Yeredor:2012gv,Weiss:2017il}.
There, the temporal covariance matrices of the sources are assumed to be known.
Let them be defined as $\mPsi_{kf}\in \C^{N\times N}$, with entries $(\mPsi)_{n,n^\prime} = \E[s_{kfn} s_{kfn^\prime}^*]$.
Then, applying maximum likelihood estimation to this problem leads to the minimization of~\eref{cost_auxiva} with the alternative definition $\mV_{kf} = \frac{1}{N} \mX_f \mPsi_{kf}^{-1} \mX_f^\H$, where $\mX = [\vx_{f1}\, \cdots\, \vx_{fN}]$~\cite{Weiss:2017il}.

In the BSS problem considered in this paper, $\mPsi_{kf}$ is unknown.
However, we can interpret the AuxIVA algorithm as solving a sequence of Gaussian separation problems by SeDJoCo.
At each iteration, we estimate the temporal covariance matrix as $\mPsi_{kf} \approx \diag(\psi(\| \bar{\vs}_{k1} \|),\ldots, \psi(\|\bar{\vs}_{kN} \|))$, where $\bar{\vs}_{kn}$ is the current estimate of source $k$, and $\psi(r) = G^\prime(r) / (2r)$.
After solving the SeDJoCo with the current value of $\mPsi_{kf}$, we update our source estimate and repeat the process.
We emphasize this is only an interpretation and that the soundness of AuxIVA comes from its derivation as an MM algorithm minimizing~\eref{cost_iva}.

\begin{table*}
  \centering
  \caption{Illustration and properties of the demixing matrix parametrization in the different update methods used for AuxIVA: IP~\cite{Ono:2011tn}, IP2~\cite{Ono:2018wk,Scheibler:2020mm}, ISS~\cite{Scheibler:2019tt}, and IPA (proposed).}
  \includegraphics{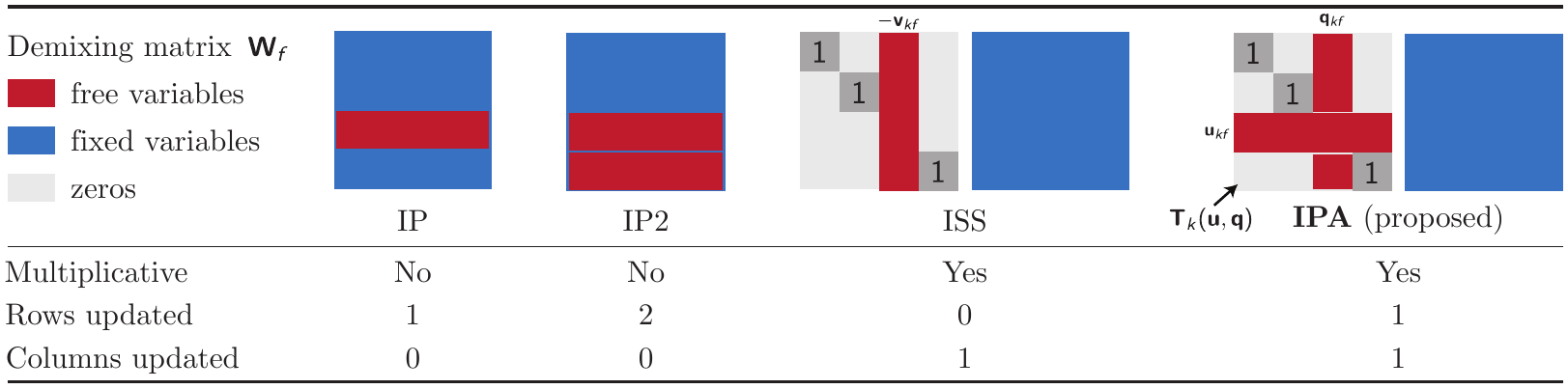}
  \tlabel{auxiva_visual_updates}
\end{table*}

\section{Iterative Projection with Adjustment}
\seclabel{ipa}

The key to make AuxIVA faster is to reduce the surrogate function by a larger amount at each iteration.
The block-coordinate descent  IP, IP2, and ISS, fix a part of the demixing matrix and minimize~\eref{cost_auxiva} over the remaining free variables.
Consequently, IP and IP2 only adjust one or two demixing filters at a time.
If the other sources are not well-separated, this might limit the progress that can be made at a single iteration.
ISS operates similarly, but implicitly, on the columns of $\mW^{-1}$~\cite{Scheibler:2020ig}.

Instead, we propose IPA, a new update that performs jointly an IP-style and an ISS-style update.
We completely re-estimate the $k$th demixing filter, and, jointly, we adjust the values of all other filters by taking a step aligned with the current estimate of source $k$.
This gives a chance for all other sources to be adjusted according to the new estimate of source $k$.
We illustrate all the parametrizations in \tref{auxiva_visual_updates}.

Assuming we have an initial value $\wh{\mW}$ for the demixing matrix, we propose the following parametrization of $\mW$ in terms of $\vu\in \C^M$ and $\vq\in \C^{M-1}$,
\begin{equation}
  \mW \gets \mT_{k}(\vu, \vq) \wh{\mW},
  \elabel{ipa_mu}
\end{equation}
with $\mT_k(\vu, \vq)$ illustrated in \tref{auxiva_visual_updates}, and defined as,
\begin{equation}
  \mT_k(\vu, \vq) = \mI + \ve_k (\vu - \ve_k)^\H + \bar{\mE}_k \vq^* \ve_k^\top,
  \elabel{ipa_matrix}
\end{equation}
with $\bar{\mE}_k$ being the $M\times (M-1)$ matrix containing all canonical basis vectors but the $k$th,
\begin{align}  
  \bar{\mE}_{k} & = \begin{bmatrix} \ve_1 & \cdots & \ve_{k-1} & \ve_{k+1} & \cdots & \ve_M \end{bmatrix}.
\end{align}
In the final MM algorithm, i.e.,~\algref{auxiva_ipa}, $\wh{\mW}$ is chosen as the estimate of the demixing matrix from the previous iteration so that \eref{ipa_mu} is an update equation.
However, without loss of generality, we can assume $\wh{\mW} = \mI$ in the derivations, since in \eref{cost_auxiva} it can be absorbed into the weighted covariance matrices $\mV_{kf}$ and some constant factors.
Note that we removed the index $f$ to lighten the notation, and because optimization of $\eref{cost_auxiva}$ can be carried out separately for different $f$.

First, note that the $m$th row of $\mT_k(\vu,\vq)$, for $m\neq k$, is $(\ve_m + q_m \ve_k)^\H$.
Thus, plugging \eref{ipa_mu} into the IVA surrogate \eref{cost_auxiva}, with a slight abuse of notation, we obtain,
\begin{multline}
  \ell_2(\vu, \vq)  = \sum_{m\neq k} (\ve_m + q_m \ve_k)^\H \mV_m (\ve_m + q_m \ve_k) \\ + \vu^\H \mV_k \vu - 2 \log|\det\mT_k(\vu, \vq)|,
    \elabel{cost_func_Tk}
\end{multline}
and we want to find the optimal values of $\vu$ and $\vq$, i.e.,
\begin{equation}
  \vu^\star, \vq^\star = \underset{\vu\in\C^M, \vq \in \C^{M-1}}{\arg \min}\ \ell_2(\vu, \vq). \\
  \elabel{argmin_func_Tk}
\end{equation}
Albeit not convex, it turns out that the solution of this optimization problem can be found efficiently.
First, we show that a closed-form solution for $\vu$ as a function of $\vq$ exists.
Then, plugging the expression for $\vu$ back in the cost function, we find that the optimal $\vq$ is given by the solution of Problem~\ref{problem:lqpqm}.
This is formalized in Theorem~\ref{theorem:opt_u_q}.
An efficient algorithm to solve Problem~\ref{problem:lqpqm} is described in the following section and the final procedure is given in \algref{auxiva_ipa}.

\begin{theorem}
  \label{theorem:opt_u_q}
  Let $\mV_1,\ldots,\mV_M$ be $M$ Hermitian positive definite matrices.
  Then, the solution of \eref{argmin_func_Tk} is as follows.
  \begin{enumerate}
    \item
      For a given $\vq$, the optimal vector $\vu^\star(\vq)$ is given by
      \begin{align}
        \vu^\star(\vq) = \frac{\mV_k^{-1}\vqt_k}{\sqrt{\vqt_k^\H \mV_k^{-1} \vqt_k}}e^{j\theta}.
        \elabel{opt_u}
      \end{align}
      where we defined $\vqt$ for convenience as
      \begin{equation}
        \vqt_k = \ve_k - \bar{\mE}_k \vq^*,
        \elabel{q_tilde}
      \end{equation}
      and $\theta\in[0, 2\pi]$ is an arbitrary phase.
    \item
      The optimal $\vq^\star$ is the solution to the following instance of Problem~\ref{problem:lqpqm},
      \begin{multline}
        \underset{\vq \in \C^{M-1}}{\min}\ 
        (\vq + \mA^{-1} \vb)^\H \mA (\vq + \mA^{-1} \vb) \\
        - \log \left((\vq - \mC^{-1}\vg)^\H \mC (\vq - \mC^{-1}\vg) + z\right)
      \end{multline}
      with
      \begin{align}
        \mA & = \diag(\ldots,\, \ve_k^\top \mV_m \ve_k,\,\ldots),\quad m \neq k, \\
        \vb & = \begin{bmatrix} \, \cdots & \ve_k^\top \mV_m \ve_m & \cdots\, \end{bmatrix}^\top,\quad m\neq k \\
        \mC & = \bar{\mE}_k^\top (\mV_k^{-1})^* \bar{\mE}_k, \\
        \vg & = \bar{\mE}_k^\top (\mV_k^{-1})^* \ve_k, \\
        z & = \ve_k^\top(\mV_k^{-1})^*\ve_k  - \vg^\H \mC^{-1} \vg.
      \end{align}
    \end{enumerate}
\end{theorem}

We note that the phase ambiguity in~\eref{opt_u} is unavoidable.
Indeed, the cost function of IVA~\eref{cost_iva} with the spherical source model~\eref{spherical} is invariant to the choice of $\theta$.
Furthermore, IVA suffers from a scale ambiguity that is usually fixed by a post-processing step~\cite{Murata:2001gb,Matsuoka:2002da}.
In practice, we always fix $\theta=0$.

\begin{algorithm}[t]
  \SetKwInOut{Input}{Input}\SetKwInOut{Output}{Output}
  \SetKw{KwBy}{by}
  \Input{$\mW$, $\mV_{1}$, $\ldots$, $\mV_{M}$}
  \Output{Updated matrix $\mW$}
  \DontPrintSemicolon

  \For{$k\gets 1$ \KwTo $M$}{
    $\mA \gets \diag(\ldots,\, \vw_{k}^\H \mV_{m} \vw_{k},\,\ldots),\quad m \neq k$\;
    $\vb \gets \begin{bmatrix} \, \cdots & \vw_{k}^\H \mV_{m} \vw_{m} & \cdots\, \end{bmatrix}^\top,\quad m\neq k$\;
    $\wt{\mV} \gets ((\mW \mV_k \mW^\H)^{-1})^*$\;
    $\mC \gets \bar{\mE}_k^\top \wt{\mV} \bar{\mE}_k$\;
    $\vg \gets \bar{\mE}_k^\top \wt{\mV} \ve_k$\;
    $z \gets \ve_k^\top \wt{\mV} \ve_k  - \vg^\H \mC^{-1} \vg$\;
    $\vq, \lambda \gets \operatorname{LQPQM}(\mA, -\mA^{-1}\vb, \mC, \mC^{-1} \vg, z)$\;
    $\vu \gets \frac{1}{\sqrt{\lambda}} \wt{\mV} (\ve_k - \bar{\mE}_k \vq^*)$\;
    $\mW \gets (\mI + \ve_k (\vu^\H - \ve_k^\top) + \bar{\mE}_k \vq^* \ve_k^\top) \mW$\;
  }
  \vspace{0.5cm}
  \caption{UpdateIPA: Update sub-routine of AuxIVA implementing IPA.}
  \label{alg:auxiva_ipa}
\end{algorithm}

\begin{proof}
  We prove the two parts of the theorem in order.

      First, let us take the complex derivative of~\eref{cost_func_Tk} with respect to $\vu^*$ according to the Wirtinger calculus~\cite{Remmert:1991vp},
      \begin{align}
        \nabla_{\vu^*} \ell_2(\vu, \vq) & = \mV_k \vu - \mT_k^{-1}(\vu, \vq) \ve_k.
      \end{align}
      Note that the derivative only exists for $\vu, \vq$ such that $\mT_k(\vu, \vq)$ is full-rank.
      Equating to zero and multiplying by $\mT_k^{-1}(\vu,\vq)$ from the left, we obtain,
      \begin{align}
          \vu^\H \mV_k \vu & = 1, \elabel{u_k_u} \\
        (\bar{\mE}_k^\top + \vq^* \ve_k^\top) \mV_k \vu & = \vzero. \elabel{x_k_u}
      \end{align}
      We observe that \eref{x_k_u} is a null space constraint.
      Adding the new equation $\ve_k^\top \mV_k \vu = \eta$, we have
      \begin{align}
        (\mI + \bar{\mE}_k \vq^* \ve_k^\top) \mV_k \vu = \eta \ve_k, 
      \end{align}
      where $\eta\in\C$ is a new variable that may be adjusted later to satisfy \eref{u_k_u}.
      Using the matrix inversion lemma, we obtain $\vu$ as a function of $\vq$ and $\eta$,
      \begin{align}
        \vu &= \eta \mV_k^{-1} (\mI + \bar{\mE}_k \vq^* \ve_k^\top)^{-1} \ve_k  \\
            & = \eta \mV_k^{-1} \left(\mI - \frac{\bar{\mE}_k \vq^* \ve_k^\top }{1 + \ve_k^\top \bar{\mE}_k \vq^*}\right)\ve_k \\
            & = \eta \mV_k^{-1} \left(\ve_k - \bar{\mE}_k \vq^*\right) = \eta \mV_k^{-1} \vqt_k,
        \elabel{Vk_u_eq}
      \end{align}
      where $\vqt_k$ is from \eref{q_tilde}, and we used the fact that $\ve_k^\top \bar{\mE}_k \vq = 0$.
      Now, we substitute~\eref{Vk_u_eq} in~\eref{u_k_u},
      \begin{align}
        \vu^\H \mV_k \vu & = |\eta|^2 \vqt_k^\H \mV_k^{-1} \mV_k \mV_k^{-1} \vqt_k = 1,
      \end{align}
      and solving for $\eta$ yields,
      \begin{align}
        \eta = e^{j\theta}\left( \vqt_k^\H \mV_k^{-1} \vqt_k \right)^{-\half}.
      \end{align}
      Together with \eref{Vk_u_eq}, this gives \eref{opt_u}.

      The proof of the second part follows from substituting $\vu^\star$ from \eref{opt_u} into the objective function \eref{cost_func_Tk}.
      \begin{enumerate}
        \item By \eref{u_k_u}, the quadratic term in $\vu$ equals one.
        \item Now, we handle the log-determinant part. In \aref{det_Tk}, we show that
          \begin{equation}
            \det(\mT_k) = \vu^\H\vqt_k.
          \end{equation}
          Substituting $\vu^\star$, we further have
          \begin{align}
            |(\vu^\star)^\H\vqt_k|
            = \left|\frac{\vqt_k^\H \mV_k^{-1} \vqt_k}{\sqrt{\vqt_k^\H \mV_k^{-1} \vqt_k}}\right|
            = \sqrt{\vqt_k^\H \mV_k^{-1} \vqt_k}.
          \end{align}
          Finally, with a little algebra, one can check that
          \begin{equation}
            \vqt_k^\H \mV_k^{-1} \vqt_k = (\vq - \mC^{-1}\vg)^\H \mC (\vq - \mC^{-1}\vg) + z.
            \nonumber
          \end{equation}
        \item As shown in \aref{quad_form}, the remaining quadratic terms can be transformed into standard quadratic form,
          \begin{multline}
            \sum_{m\neq k} (\ve_m + q_m \ve_k)^\H \mV_m (\ve_m + q_m \ve_k) \\
              = (\vq + \mA^{-1} \vb)^\H \mA (\vq + \mA^{-1} \vb) - \vb^\H \mA^{-1} \vb + \vone^\top \vc,
              \nonumber
          \end{multline}
          where $c_m = \ve_m^\top \mV_m \ve_m$, and $\vone$ is the all one vector.
      \end{enumerate}
      Removing the constant terms yields the proof.
\end{proof}

\section{Log-quadratically Penalized Quadratic Minimization}
\seclabel{LQPQM}

We will now provide an efficient algorithm to compute the solution of Problem~\ref{problem:lqpqm}.
It is interesting to take a look at the landscape of one instance of the 2D problem as shown in \ffref{lqpqm_landscape}.
First, let us give an intuitive and informal description of the problem.
The quadratic term of the objective forms the familiar bowl shape, and the log-quadratic term appears like someone pinched and pulled up the "fabric" of the cost function in one point.
The location of the "pinch", described by offset vectors $\vb$ and $\vd$, as well as the offset $z$ in the log, creates different patterns of stationary points. 
In the 2D case of \ffref{lqpqm_landscape}, we observe two "bowls", separated by a kind of ridge, which is due to the log-quadratic term.
There are in fact only a finite number of stationary points, five in \ffref{lqpqm_landscape}, to be precise.
In the rest of this section, we will make precise this intuitive description, and give a procedure to find the global minimum.

Since $\mA$ (in Problem~\ref{problem:lqpqm}) is Hermitian positive definite, it has a Cholesky decomposition, which can be inverted.
This allows to consider the following alternative form of LQPQM instead.  \begin{problem}[LQPQM alternative form] \label{problem:lqpqm_alt}
  Let $\mU \in \C^{d\times d}$ be Hermitian positive semi-definite, and $\vv \in \C^d$.
  \begin{equation}
    \underset{\vy\in\C^{d}}{\min}\ \vy^\H\vy - \log \left((\vy + \vv)^\H \mU (\vy + \vv) + z\right)
      \tag{P2}\elabel{lqpqm_eq}
  \end{equation}
\end{problem}
The two problems are equivalent.
To obtain Problem~\ref{problem:lqpqm_alt} from Problem~\ref{problem:lqpqm}, let $\mG$ be the Cholesky decomposition of $\mA$, such that $\mA = \mG^\H \mG$, and introduce the substitutions
\begin{align}
  \vy = \mG (\vx - \vb), \quad \mU = \mG^{-\H} \mC \mG^{-1}, \quad \vv = \mG (\vb - \vd).
\end{align}
The objective function of~\eref{lqpqm_eq} is bounded from below and takes its minimum at a finite value (see~\aref{bounded_below}), so that we may attempt to minimize it.
Then, given an optimizer $\vy^\star$ of Problem~\ref{problem:lqpqm_alt}, the corresponding optimizer of Problem~\ref{problem:lqpqm} is
\begin{equation}
  \vq^\star = \mG^{-1} \vy^\star + \vb.
\end{equation}

The next two theorems fully characterize the solution of Problem~\ref{problem:lqpqm} and~\ref{problem:lqpqm_alt}.
Theorem~\ref{thm:special_case} handles the case when the offset vector $\vv$ is zero (or $\vb = \vd$ in Problem~\ref{problem:lqpqm}).
There, the solution can be obtained from the eigendecomposition of $\mU$.
Note that the eigendecomposition of $\mU$ is equivalent to the generalized eigendecomposition of $\mA$ and $\mB$.
When $\vv\neq \vzero$, the solution can be computed by solving a non-linear equation as explained in Theorem~\ref{thm:general_case}.
An algorithmic instantiation of these two theorems is provided by \algref{lqpqm}.

\begin{theorem}[Special Case, $\vv = \vzero$]
  \label{thm:special_case}
  The global minimum of \eref{lqpqm_eq} is characterized as follows.
  Let $\varphi_1 \leq \ldots \leq \varphi_d$ be the eigenvalues of $\mU$, and $\vsigma_1,\ldots, \vsigma_d$, the corresponding eigenvectors.
  \begin{enumerate}
    \item If $z \geq \varphi_d$, then $\vy^\star = \vzero$ is the unique global minimizer.
    \item If $z < \varphi_d$, the minimizer is given by
      \begin{equation}
      \vy^\star =  e^{j\theta} \sqrt{\frac{\varphi_d - z}{\tilde{\vsigma}^\H \mU \tilde{\vsigma}}} \tilde{\vsigma},
      \end{equation}
      where $\theta\in[0, 2\pi]$ is an arbitrary phase.
      If $\varphi_d > \varphi_{d-1}$, the global minimizer is unique (up to the phase $\theta$) and given by $\tilde{\vsigma} = \vsigma_d$.
      If the largest eigenvalue has multiplicity $k$, then any linear combination $\tilde{\vsigma}$ of $\vsigma_{d-k},\ldots, \vsigma_d$ is a global minimizer.
  \end{enumerate}
\end{theorem}
\begin{theorem}[General Case, $\vv \neq \vzero$]
  \label{thm:general_case}
  Let $\mU = \mSigma \mPhi \mSigma^\H$ be the eigendecomposition of $\mU$, with $\mPhi = \diag(\varphi_1,\ldots, \varphi_d)$, where $\varphi_1 \leq \ldots \leq \varphi_d$ are the eigenvalues of $\mU$.
  Then, the unique global minimizer of \eref{lqpqm_eq} is
  \begin{equation}
    \vy^{\star} = (\lambda^\star \mI - \mU)^{-1} \mU \vv
  \end{equation}
  where $\lambda^\star$ is the largest root of the function $f\,:\,\R_+ \to \R$,
  \begin{equation}
    f(\lambda) = \lambda^2 \sum_{m \in \calS} \frac{\varphi_m |\tilde{v}_m|^2}{(\lambda - \varphi_m)^2} - \lambda + z,
    \elabel{f_lambda}
  \end{equation}
  where $\tilde{v}_m$ are the coefficients of the vector $\vvt = \mSigma^\H \vv$, and $\calS$ is the common support of $\tilde{v}$ and the eigenvalues,
  \begin{equation}
    \calS = \{ m \, :\, \varphi_m |\tilde{v}_m|^2 \neq 0 \}.
    \elabel{f_support}
  \end{equation}
  Furthermore, the largest root is the unique root located in the interval $(\max(\varphi_{\max}, z), +\infty)$, where $\varphi_{\max} = \max_{m\in\calS} \ \varphi_m$.
  In this interval, $f(\lambda)$ is strictly decreasing.
\end{theorem}

Because the optimal $\lambda$ is restricted to an interval where $f(\lambda)$ is strictly decreasing, we may use a root finding algorithm to compute it efficiently, as explained in~\sref{root_finding}.
The complete procedure for LQPQM is described in \algref{lqpqm}.
\algref{solve_secular_equation} is the sub-routine solving the equation $f(\lambda) = 0$.

\begin{figure}
  \centering
  \includegraphics[width=\linewidth]{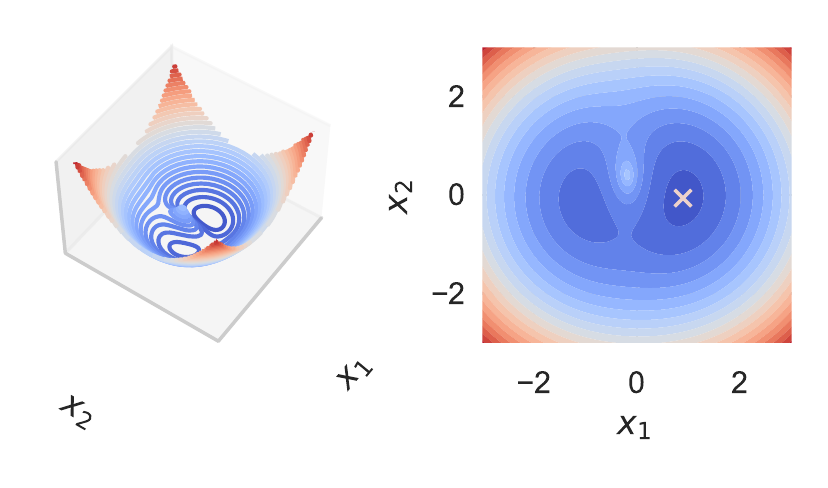}
  \caption{The loss landscape of an instance of the 2D LQPQM shown in 3D (left) and 2D (right) contour plots. The global minimum is indicated by an $\times$ on the right figure.}
  \flabel{lqpqm_landscape}
\end{figure}

\begin{algorithm}[t]
  \SetKwInOut{Input}{Input}\SetKwInOut{Output}{Output}
  \SetKw{KwBy}{by}
  \Input{$\mA$, $\vb$, $\mC$, $\vd$, $z$}
  \Output{$\vq$, $\lambda$, solution to Problem~\ref{problem:lqpqm}}
  \DontPrintSemicolon
  $\mG \gets \operatorname{Cholesky}(\mA)$\;
  $\mU \gets \mG^{-\H} \mC \mG^{-1}$\;
  $\mPhi, \mSigma \gets \operatorname{EigenValueDecomposition}(\mU)$\;
  \If{$\vb = \vd$}{
    \If{$z \geq \varphi_d$}{
      $\lambda \gets z$\;
      $\vy \gets \vzero$\;
    }
    \Else{
      $\lambda \gets \varphi_d$\;
      $\vy \gets \sqrt{\frac{\varphi_d - z}{\vsigma_d^\H\mU \vsigma_d}}\vsigma_d$\;
    }
  }
  \Else{
    $\vvt \gets \mSigma^\H \mG(\vb - \vd)$\;
    $\mu \gets \operatorname{SolveEquation}\left(\frac{\mPhi}{\varphi_{\max}}, \frac{\vvt}{\varphi_{\max}}, \frac{z}{\varphi_{\max}}\right)$\;
    $\lambda \gets \mu\, \varphi_{\max}$\;
    $\vy \gets \mSigma (\lambda\mI - \mPhi)^{-1} \mPhi \vvt$\;
  }
  $\vx \gets \mG^{-1} \vy + \vb$\;
  \vspace{0.5cm}
  \caption{LQPQM}
  \label{alg:lqpqm}
\end{algorithm}

\subsection{Proof of Theorem~\ref{thm:special_case}}
\seclabel{proof_special_case}

The special case, $\vv = \vzero$, leads to the simpler problem,
\begin{equation}
  \underset{\vy \in \mC^d}{\min}\ \vy^\H \vy -\log(\vy^\H \mU \vy + z).
\end{equation}
Equating the gradient to zero, and adding an extra non-negative variable $\lambda\geq 0$, we obtain the following first order necessary optimality conditions,
\begin{equation}
\left\{
    \begin{array}{rl}
        \mU \vy & = \lambda \vy, \\
        \lambda & = \vy^\H \mU \vy + z.
    \end{array}
\right.
\elabel{sc1_nec_cond}
\end{equation}
Solutions to this system of equations are stationary points.
\begin{itemize}
  \item The trivial solution to \eref{sc1_nec_cond}: $\lambda = z$, $\vy = \vzero$.
  \item The eigenvalue/vectors of $\mU$ also provide the solutions $\lambda = \varphi_i$, $ \vy = \eta \vsigma_i$, where $\eta$ is an unknown scale, for all $i$.
    Replacing in the second equation, we obtain
    \begin{align}
      \varphi_i = |\eta|^2 \vsigma_i^\H \mU \vsigma_i + z.
    \end{align}
    For all $\varphi_i \geq z$, this equation admits the solution
    \begin{align}
      \eta = e^{j\theta} \sqrt{\frac{\varphi_d - z}{\vsigma_i^\H \mU \vsigma_i}},
    \end{align}
    where $\theta\in[0, 2\pi]$ is an arbitrary phase.
\end{itemize}
From \eref{sc1_nec_cond}, we obtain $\vy^\H \vy = (\lambda - z) / \lambda$.
Together with the second equation in \eref{sc1_nec_cond}, this allows to rewrite the objective as a function of $\lambda$,
\begin{equation}
  g(\lambda) = -\log \lambda + \frac{\lambda - z}{\lambda}.
\end{equation}
The derivative is
\begin{equation}
  g^\prime(\lambda) = \frac{z - \lambda}{\lambda^2},
\end{equation}
and $g(\lambda)$ is thus decreasing for $\lambda > z$.
Thus, if $\varphi_d \geq z$, the solution is given by the largest eigenvector (or eigenvectors if the multiplicity of the largest eigenvalue is more than one).
Otherwise, the optimum is zero.
\hfill $\square$

\subsection{Proof of Theorem~\ref{thm:general_case}}

Equating the gradient of the objective of \eref{lqpqm_eq} with respect to $\vy^*$ to zero, we obtain the following equation,
\begin{equation}
  \vy  - \frac{\mU (\vy + \vv)}{(\vy + \vv)^\H \mU (\vy + \vv) + z} = \vzero.
\end{equation}
As in the previous section, we isolate the quadratic term in a second equation by adding the non-negative variable $\lambda \geq 0$, and obtain the following first order optimality conditions,
\begin{equation}
  \left\{
    \begin{array}{rl}
      \mU (\vy + \vv) & = \lambda \vy, \\
      \lambda & = (\vy + \vv)^\H \mU (\vy + \vv) + z.
    \end{array}
  \right.
  \elabel{opt_cond}
\end{equation}
Solving the first equation of~\eref{opt_cond} for $\vy$, we obtain a solution as a function of $\lambda$,
\begin{equation}
  \vy(\lambda) = (\lambda \mI - \mU)^{-1} \mU \vv.
  \elabel{opt_y}
\end{equation}
Switching to the eigenbasis of $\mU$ and substituting~\eref{opt_y} into the second equation of~\eref{opt_cond} leads to
\begin{align}
\lambda & = \| \mPhi^{\half}  ((\lambda \mI - \mPhi)^{-1} \mPhi + \mI) \vvt\|^2 + z \\
        & = \sum_{m=1}^d \varphi_m \left| \left(\frac{\varphi_m}{\lambda - \varphi_m} + 1\right) \tilde{v}_m \right|^2 + z \\
        & = \lambda^2 \sum_{m \in \calS} \frac{\varphi_m |\tilde{v}_m|^2}{(\lambda - \varphi_m)^2} + z.
\end{align}
This gives us the necessary condition that $f(\lambda) = 0$ for any stationary point of \eref{lqpqm_eq}.
Now this equation may have multiple roots, so we need to find the one with the lowest value of the objective.
It turns out that the value of the objective can also be written as the following function of $\lambda$ only,
\begin{equation}
  g(\lambda) = 1 - \sum_{m\in\calS} \frac{\varphi_m |\tilde{v}_m|^2}{(\lambda - \varphi_m)} - \frac{z}{\lambda} - \log \lambda.
  \elabel{obj_g_lambda}
\end{equation}
The proof is provided in~\aref{lqpqm_stationary_points}, Lemma~\ref{lemma:lqpqm_obj_val}.
Thus, the optimal $\lambda$ is the solution to the following problem,
\begin{equation}
    \underset{\lambda \in \R_+}{\min}\ g(\lambda),\quad
    \text{subject to}\ f(\lambda) = 0.
  \elabel{p3_lambda}\tag{P3}
\end{equation}
where $f(\lambda)$ is defined in \eref{f_lambda}.
In \ffref{secular}, we show the functions $g(\lambda)$ and  $f(\lambda)$ for the instance of LQPQM of \ffref{lqpqm_landscape}.
This new problem is highly non-linear and the objective is not even continuous.
However, we can show that $f(\lambda)$ only has a finite number of roots and that the largest, $\lambda^\star$, has the minimum value of the objective among them.
In particular, we prove in \aref{lqpqm_stationary_points} the following about $f(\lambda)$ and its zeros.
\begin{enumerate}
  \item Lemma~\ref{lemma:f_zeros}: The largest zero of $f(\lambda)$ is the unique zero located in $(\max(z, \varphi_d), +\infty)$.
    Furthermore, $f(\lambda)$ is strictly decreasing in this interval.
  \item Lemmas~\ref{lemma:f_zeros_between_ev} and~\ref{lemma:f_zeros_across_ev}: The objective value is decreasing for increasing zeros of $f(\lambda)$, i.e., if $\lambda_1 \leq \ldots \leq \lambda^\star$ are all solutions of $f(\lambda) = 0$, then $g(\lambda_1) \geq \ldots \geq g(\lambda^\star)$.
\end{enumerate}
Thus, $\vy(\lambda^\star)$ is the global minimizer of \eref{lqpqm_eq}.
\hfill$\square$

\begin{figure}
  \centering
  \includegraphics[width=0.8\linewidth]{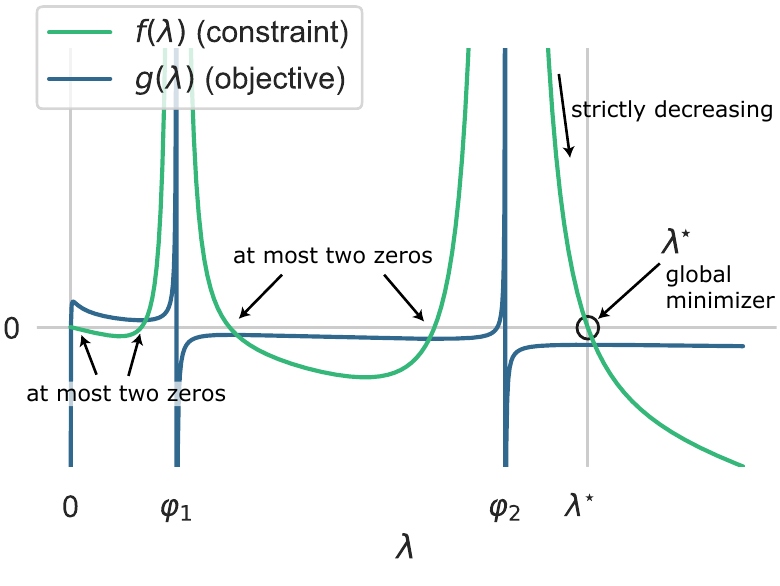}
  \caption{The function $f(\lambda)$ corresponding to the 2D LQPQM in
  \ffref{lqpqm_landscape} and its objective value $g(\lambda)$. The location of the zeros of $f(\lambda)$ are pointed out and correspond to stationary points of the objective. The optimal $\lambda^\star$ is the largest root of $f(\lambda)$ and corresponds to the global minimum.}
      \flabel{secular}
\end{figure}

\begin{algorithm}[t]
  \SetKwInOut{Input}{Input}\SetKwInOut{Output}{Output}
  \SetKw{KwBy}{by}
  \Input{$\mPhi$, $\vvt$, $z$}
  \Output{Largest zero of $f(\lambda)$}
  \DontPrintSemicolon

  $\lambda \gets \operatorname{InitCubicPoly}(\varphi_{\max}, \tilde{v}_{\max}, z)$\;
  $\lambda \gets \max(\lambda, z)$\;

  \While{$|f(\lambda)| > \epsilon$}{
    $\mu \gets \lambda - \frac{f(\lambda)}{f^\prime(\lambda)}$\;
    \If{$\mu > \varphi_{\max}$}{
      $\lambda \gets \mu$\;
    }
    \Else{
      $\lambda \gets \frac{\varphi_{\max} + \lambda}{2}$\;
    }
  }

  \vspace{0.5cm}
  \caption{SolveEquation. The routine to compute the largest root of $f(\lambda)$. The sub-routine InitCubicPoly returns the largest real root of the cubic polynomial \eref{cubic_poly}.}
  \label{alg:solve_secular_equation}
\end{algorithm}

\subsection{Root Finding}
\seclabel{root_finding}

The solution to the general problem \eref{lqpqm_eq} is given by the largest root of $f(\lambda)$, from \eref{f_lambda}.
We have shown that the root is in $(\max(\varphi_{\max}, z), +\infty)$, and we can thus use a root finding algorithm to find it.
We propose to use the Newton-Raphson method protected by bisection on the left, as described in \algref{solve_secular_equation}.
With a good initial point, this method converges in just a few iterations.
We describe in~\aref{root_finding_init} an initialization procedure based on solving a cubic polynomial.

When the eigenvalues are large, computation of $(\lambda - \varphi_m)^{-2}$ may lead to an overflow, jeopardizing the algorithm.
Instead, we consider
\begin{equation}
  \hat{f}(\mu) = \frac{1}{\varphi_{\max}} f(\varphi_{\max} \mu) = \mu^2 \sum\nolimits_m \frac{\hat{\varphi}_m |\hat{v}_m|^2}{(\mu - \hat{\varphi}_m)^2} - \mu + \hat{z},
\end{equation}
with $\hat{\varphi}_m = \varphi_m / \varphi_{\max}$, $\hat{v}_m = \tilde{v}_m / \varphi_{\max}$, and $\hat{z} = z / \varphi_{\max}$.
We can find the largest root of $\hat{f}(\mu) = 0$, $\mu^*$, with \algref{solve_secular_equation}.
Then, the largest root of $f(\lambda)$ is $\lambda^* = \varphi_{\max}\,\mu^*$.

\subsection{Computational Complexity}

The computational complexity of a single iteration of AuxIVA with IP, IP2, and IPA is dominated by the computation of the $M$ weighted covariance matrices $\mV_k$ of~\eref{auxiliary_variable_update}, which has order $O(M^3 N F)$.
The other operations required for each algorithm, per iteration, source, and mixture, are as follows.
IP requires one matrix inversion for a total of $O(M^4 F)$.
IP2 requires two matrix inversions and one generalized eigenvalue decomposition for a total of $O(M^4 F)$.
IPA requires one matrix inversion, two matrix-matrix multiplications, one eigenvalues decomposition for a total of $O(M^4 F)$.
The root finding requires $O(M F)$ per source and iteration of \algref{solve_secular_equation}, and thus does not increase the complexity.
Since in general $N \gg M$, the overall complexity is $O(M^3 N F)$ for these three methods.
AuxIVA with ISS has the particularity that an efficient algorithm fusing the computation of $\mV_k$ and the update~\eref{iss_update_2} exists, with complexity $O(M^2 N F)$~\cite{Scheibler:2020ig}.
For reference, the NCG algorithm for SeDJoCo has complexity $O(M^5)$~\cite{Yeredor:hr,Yeredor:2012gv} (for a single mixture).

\section{Numerical Experiments}
\seclabel{experiments}

\subsection{Solving Random SeDJoCo Problems}
\seclabel{random_sedjoco}

Our first experiment compares the performance of the different methods to solve SeDJoCo only, i.e., the minimization of the surrogate function~\eref{cost_auxiva}.
We generate sets of $M$ random Hermitian matrices with zero-mean unit-variance normally distributed coefficients, and make them positive definite by making their eigenvalues positive.
We initialize the algorithms with $\mW=\mI$ and run 1000 iterations of IP, ISS, IP2, NCG, and IPA+NCG~\cite{Yeredor:hr,Yeredor:2012gv}.
IPA+NCG is NCG initialized by the IPA solution after the SeDJoCo residual is less than $10^{-5}$, where the residual is defined as,
\begin{align}
  \left\| \mW \begin{bmatrix} \mV_1\vw_1 & \cdots & \mV_M \vw_M \end{bmatrix} - \mI \right\|_F^2.
  \elabel{sedjoco_residual}
\end{align}

\ffref{experiment_sedjoco} shows the evolution of the median SeDJoCo residual~\eref{sedjoco_residual} and the median value of the surrogate cost function~\eref{cost_auxiva}.  
The SeDJoCo residual plateaued for all algorithms around $10^{-30}$.
IPA+NCG, followed by NCG, get there the fastest.
The other algorithms are from fastest to slowest, IPA, IP2, IP, and ISS.
In terms of cost, IPA, IP2, IP, and ISS, in that order are the fastest.
NCG is the slowest, and seems to settle to a higher final median cost, indicating that it might end up in worse local minima.
The iteration where IPA+NCG switches to NCG is visible as the cost function starts increasing before decreasing again.
This demonstrates that NCG might not be an appropriate choice for AuxIVA.

For an efficient MM algorithm, the surrogate minimization step should decrease the cost function as much as possible.
In~\tref{surrogate_cost_decrease}, we compare how much the cost function decreases in the first two iterations for IP, ISS, IP2, and IPA.
It shows the median ratio of the cost decrease of one algorithm to that of IPA, and the median is taken over all samples.
We see that at the first iteration, the decrease of ISS, IP, and IP2 are approximately \SI{45}{\percent}, \SI{80}{\percent}, and \SI{90}{\percent}, respectively, that of IPA.
This is an indication that we can expect \algref{auxiva} to converge faster when using IPA.

\begin{figure}
  \centering
  \includegraphics[width=\linewidth]{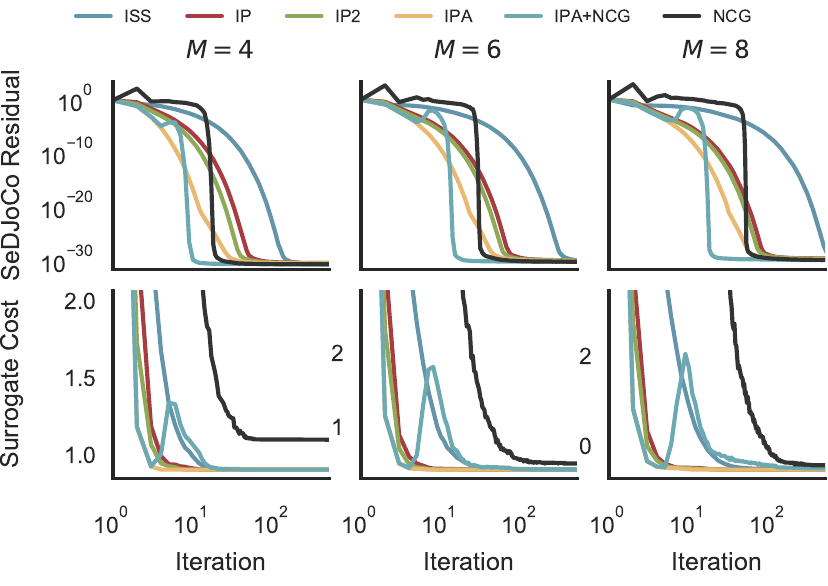}
  \caption{Performance of the different algorithms for solving random SeDJoCo problems. Top, the median SeDJoCo residual~\eref{sedjoco_residual} as a function of the number of iterations. Bottom, the median value of the associated cost function~\eref{cost_auxiva}.}
  \flabel{experiment_sedjoco}
\end{figure}

\begin{table}
  \centering
  \caption{Decrease of the surrogate cost function~\eref{cost_auxiva} for different update algorithms relative to that of IPA for the first two iterations.}
  \begin{tabular}{@{}lrrrrrr@{}}
    \toprule
& \multicolumn{2}{c}{$M=4$} & \multicolumn{2}{c}{$M=6$} & \multicolumn{2}{c}{$M=8$} \\
    \cmidrule(l{2pt}r{2pt}){2-3} \cmidrule(l{2pt}r{2pt}){4-5}\cmidrule(l{2pt}r{2pt}){6-7}
Iter.&        1 &         2 &        1 &      2 &        1 &       2 \\
\midrule
    ISS     &      46\% &       74\% &      44\% &    65\% &      44\% &     61\% \\
    IP      &      78\% &       95\% &      81\% &    96\% &      83\% &     97\% \\
    IP2     &      90\% &       98\% &      89\% &    98\% &      90\% &     98\% \\
    IPA     &     100\% &      100\% &     100\% &   100\% &     100\% &    100\% \\
    \bottomrule
  \end{tabular}
  \tlabel{surrogate_cost_decrease}
\end{table}

\begin{figure*}
  \centering
  \includegraphics[width=\linewidth]{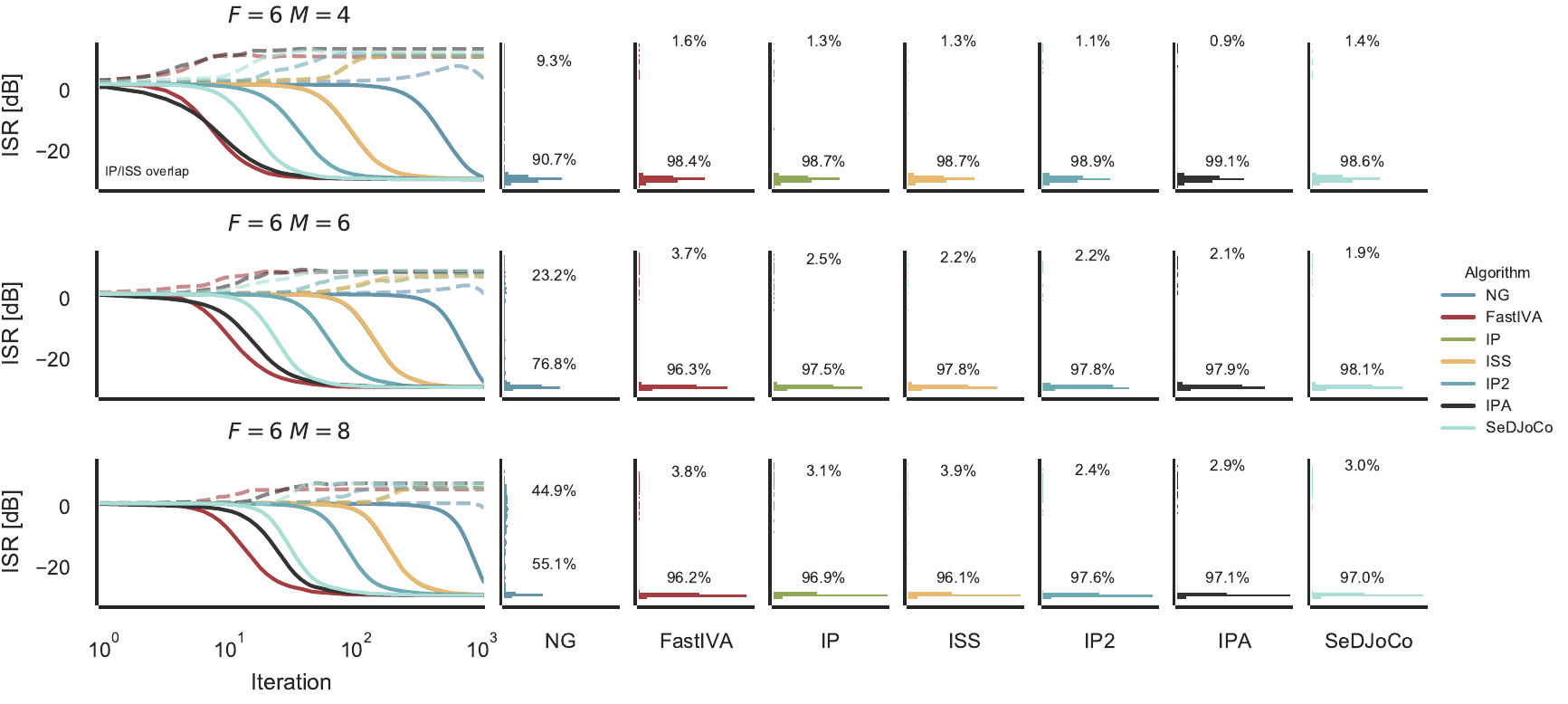}
  \caption{%
    The left-most figure shows the median ISR as a function of the iterations.
    We plot separately the median of cases converging to the true demixing matrix (plain line), defined as $\text{ISR} < \SI{-10}{\decibel}$, or a spurious solution (dashed line).
    Note that IP and ISS tightly overlap.
    The figures on the right show histograms of the distribution of the ISR for the different algorithms.
    We write the percentage of cases where the true demixing matrix is found on the bottom, and of spurious solutions on top, next to the upper mode.
    The three rows are for $M=4,6,8$ channels.
    The number of parallel mixtures is always $F=6$.
  }
  \flabel{iva_synthetic_results}
\end{figure*}

\subsection{Separation of Synthetic Mixtures}
\seclabel{synthetic_iva}

We investigate the performance of AuxIVA to separate synthetic mixtures following the complex spherically symmetric Laplace distribution.
Complex SCVs following this distribution are generated as follows,
\begin{align}
  \check{\vs}_{kn} = z_{kn} \frac{\vv_{kn}}{\| \vv_{kn} \|},\quad 
  \begin{array}{r@{\,}l}
    \vv_{kn} & \sim \calC \calN(\vzero, \mI_F), \\
    z_{kn} & \sim \operatorname{Gamma}(2 F, 1),
  \end{array}
\end{align}
with $\vv_{kn} \in \C^F$ and $z_{kn}\in \R$, $z_{kn} \geq 0$, independent\footnote{One can see that the norm of spherically symmetric Laplacian vectors should be Gamma distributed by changing to spherical coordinates and marginalizing out the direction components. The scale parameter is $2F$ for complex-valued vectors since they have twice as many components as real-valued vectors of the same dimension.}.
The coefficients of the mixing matrices are drawn independently at random from the standard complex normal distribution.
In this case, the contrast function is $G(r) = r$.
We draw at random 1000 datasets with $F=6$ mixtures, $M=4,6,8$ channels, and $N=5000$ samples.
We compare AuxIVA with IP, ISS, IP2, and IPA, as well as the natural gradient (NG) method with step size 0.3~\cite{Hiroe:2006ib,Kim:2006ex}, and the fixed-point algorithm FastIVA~\cite{Lee:2007ct} with symmetric decorrelation~\cite{Hyvarinen:ek}.
In addition, we also investigate the performance of AuxIVA when~\eref{cost_auxiva} is solved up to a stationary point at every iteration.
For this purpose, we run IPA until~\eref{sedjoco_residual} is less than $10^{-20}$ before updating the auxiliary variables~\eref{auxiliary_variable_update}.
We did not find that decreasing~\eref{sedjoco_residual} more helped.
We denote this algorithm AuxIVA-SeDJoCo (just ``SeDJoCo'' in figures and tables).
All algorithms use the same contrast function.
They are all initialized by principle component analysis (PCA) of the parallel mixtures, and run for 1000 iterations.
We measure the convergence to the true separating solution with the interference-to-signal ratio (ISR).
Let $\mW_f$ and $\mA_f$ be the estimated demixing and true mixing matrices, respectively.
Then, the ISR is defined as,
\begin{align}
  \text{ISR} = \underset{\pi \in \calP}{\min} \frac{1}{F (M^2 - M)} \sum_f \sum_m \sum_{k\neq \pi(m)} \frac{|(\mW_f \mA_f)_{mk}|^2}{|(\mW_f \mA_f)_{m\pi(m)}|^2},
  \elabel{ISR}
\end{align}
where $\calP$ is the set of permutations over $\{1,\ldots,M\}$.

\ffref{iva_synthetic_results} shows the results of the experiment.
We investigate the probability of success, defined as convergence to a point with $\text{ISR} < \SI{-10}{\decibel}$, and the speed of convergence.
All methods succeed in more than \SI{95}{\percent} of cases, with the exception of NG.
However, NG is not fully converged in all cases after 1000 iterations and thus the values indicated in the figure are not representative of its final performance.
Methods using IPA and IP2 have the highest success rate, over \SI{99}{\percent} for IPA and $M=4$.
FastIVA fails more often than other methods, possibly due to the orthogonality constraints on the demixing matrix.
For all algorithms, the success rate decreases with the number of channels.
This seems natural as the probability of a permutation occurring also increases.
FastIVA converges the fastest, followed by AuxIVA-IPA, -SeDJoCo, -IP2, -IP/ISS, and NG, in this order.
These latter algorithms require approximately, $1.5\times$, $2\times$, $5\times$, $12\times$, and $50\times$, respectively, more iterations than FastIVA for convergence of the ISR.
\tref{iva_synthetic_table} details the median number of iterations needed until convergence of the cost function~\eref{cost_iva} and ISR.
Interestingly, we note that AuxIVA-ISS is considerably better than its performance in the previous experiment led us to believe.
We conjecture this may be due to the difference between the random $\mV_{kf}$ matrices of \sref{random_sedjoco}, and those resulting from AuxIVA in this experiment.
After a few iterations of AuxIVA, the matrices $\mV_{kf}$ are already very close to satisfying~\eref{sedjoco_residual}.
We also note that running sub-iterations of IPA, in AuxIVA-SeDJoCo, does not seem to have a positive effect on the convergence speed or the success rate.
However, this involves the nested iterative optimization of two non-convex functions, and a detailed analysis is beyond the scope of this work.

\begin{table}
  \centering
  \caption{Median number of iterations until the IVA cost~\eref{cost_iva} (left) and ISR~\eref{ISR} (right) reductions are less than $10^{-3}$ and \SI{0.1}{\decibel}, respectively. The natural gradient (NG) does not always converge within 1000 iterations ($+$).}
  \begin{tabular}{@{}lrrrrrr@{}}
    \toprule
    $F=6$ & \multicolumn{3}{c}{IVA Cost} & \multicolumn{3}{c}{ISR} \\
    \cmidrule(l{2pt}r{2pt}){2-4} \cmidrule(l{2pt}r{2pt}){5-7}
    $M$ &            4 &     6 &     8 &                4 &    6 &     8 \\
    \midrule
    NG      &          729 &   1k+ &   1k+ &              567 &  828 &   1k+ \\
    ISS     &          143 &   201 &   256 &              115 &  166 &   215 \\
    IP      &          142 &   200 &   255 &              113 &  165 &   216 \\
    IP2     &           65 &    96 &   125 &               49 &   77 &   103 \\
    SeDJoCo &           28 &    37 &    46 &               21 &   30 &    37 \\
    IPA     &           20 &    29 &    39 &               14 &   22 &    31 \\
    FastIVA &           11 &    16 &    19 &               10 &   14 &    17 \\
    \bottomrule
  \end{tabular}
  \tlabel{iva_synthetic_table}
\end{table}

\subsection{Separation of Convolutive Speech Mixtures}

In the last experiment, we consider the practical application of IVA to the separation of convolutive speech mixtures recorded by a microphone array.
The experiment is done on simulated reverberant speech mixtures and the performance is evaluated in terms of scale-invariant signal-to-distortion and signal-to-interference ratios (SI-SDR and SI-SIR, respectively)~\cite{LeRoux:2018tq}.
SI-SDR measures how much the target signal is degraded, while SI-SIR indicates how much of the other sources remains.
High SI-SDR indicates both good separation and high quality.
High SI-SIR indicates good separation, but not necessarily preservation of the target source.
They are defined as follows.
Let $\mS \in \R^{T\times M}$ be the matrix containing the $M$ time-domain groundtruth reference signals in its columns.
Let $\hat{\vs}\in\R^T$ be the estimated signal, and $\vs$ one of the columns of $\mS$.
Then, the definition is as follows,
\begin{align}
  \text{SI-SDR}(\vs, \hat{\vs}) = \frac{\| \alpha \vs\|^2}{\|\alpha \vs - \hat{\vs} \|^2}, \
  \text{SI-SIR}(\vs, \hat{\vs}) = \frac{\| \alpha \vs \|^2}{\|\mS \vb\|^2}
\end{align}
where
\begin{align}
  \alpha = \frac{\hat{\vs}^\top \vs}{\|\vs\|^2}, \quad \text{and} \quad \vb = (\mS^\top \mS)^{-1} \mS^\top (\alpha \vs - \hat{\vs}).
\end{align}
The final SI-SDR and SI-SIR values are computed for the permutation of the $M$ estimated sources maximizing the latter.
We also use the $\Delta$SI-SDR and $\Delta$SI-SIR defined as the difference of these metrics applied to the separated and mixture signals.
In this experiment, we use the groundtruth noiseless reverberant signals as reference.

\subsubsection{Convolutive BSS in the Frequency Domain}
\seclabel{convolutive_mixtures}

Microphones in reverberant environment record a real-valued convolutive mixture of all the sources present in the scene,
\begin{align}
  \tilde{x}_{m}[t] = \sum_{k=1}^M \sum_{\ell=0}^{L-1} \tilde{a}_{mk}[\ell] \tilde{s}_{k}[t - \ell] + b_{m}[t]\in\R,
  \elabel{td_recording}
\end{align}
where $t\in\Z$ is the sampled time index, $\tilde{s}_{k}$ is the signal of source $k$, $\tilde{a}_{mk}$ is the $L$-taps impulse response between source $k$ and microphone $m$, and $b_{m}[t]$ is the uncorrelated microphone noise signal.
The time-domain recordings from~\eref{td_recording} can be transformed to time-frequency representation by the short-time Fourier transform (STFT)~\cite{Allen:1977in}.
The STFT is applied by splitting the time-domain signal into overlapping blocks, called frames, multiplying them by a window function, and applying the discrete Fourier transform.
The STFT representation is complex-valued, but the input signal being real-valued, it is conjugate symmetric along the frequency axis.
Assuming the frame size is sufficiently longer than the impulse response $\tilde{a}_{mk}$, the STFT of $\tilde{x}_{m}$ is approximately equal to the signal model~\eref{mixture_model} with $[(\mA_1)_{mk},\ldots,(\mA_F)_{mk}]^\top$ being the DFT of $\tilde{a}_{mk}[\ell]$, and with an extra noise term.
After performing IVA on the STFT signal as described in~\sref{background}, the STFT can be inverted to obtain the separated sources in the time domain~\cite{Griffin:1984cq}.
This step also introduces a small error due to the circular convolution inherent to the DFT.
The effect of this error is made negligible by the use of an appropriate synthesis window and sufficient overlap between the frames~\cite{Griffin:1984cq}.

\subsubsection{Setup}

\begin{figure}
  \centering
  \includegraphics[width=\linewidth]{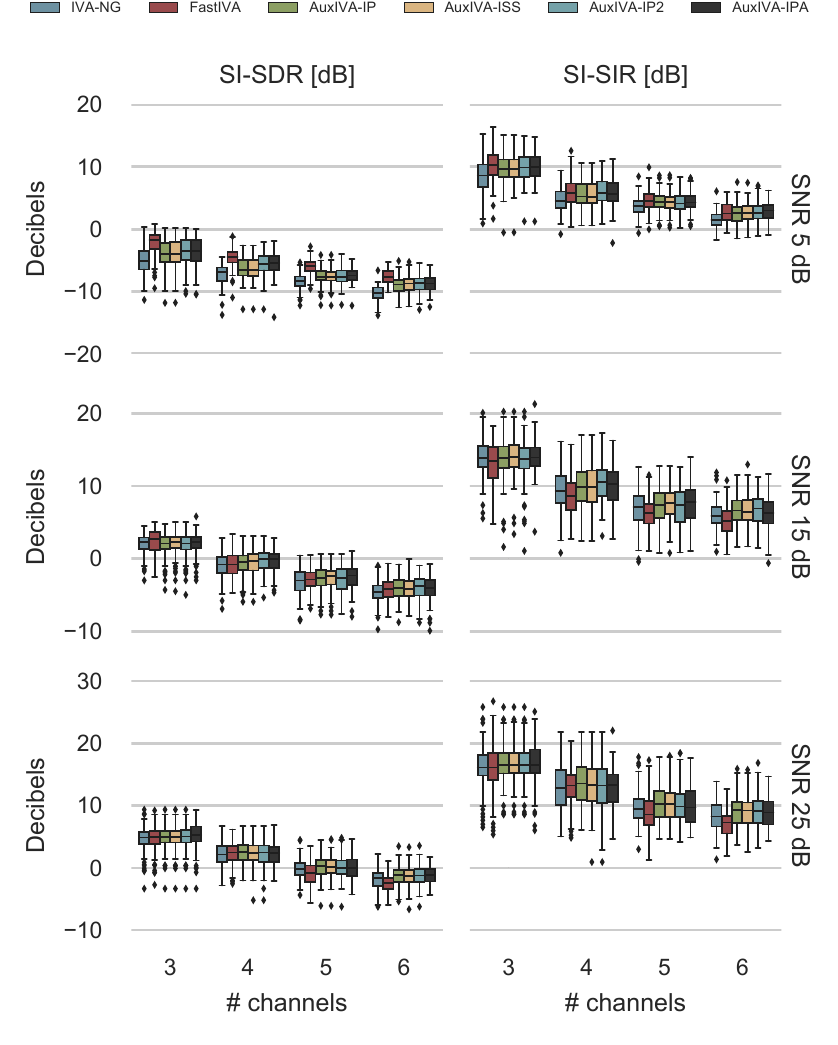}
  \caption{Box-plots of the final SI-SDR (left) and SI-SIR (right) values after a hundred iterations. From top to bottom, the SNR is \SI{5}{\decibel}, \SI{15}{\decibel}, and \SI{25}{\decibel}. In subplots, from left to right, the number of sources goes from three to six.}
  \flabel{sep_final_value}
\end{figure}

\begin{figure*}
  \centering
  \includegraphics{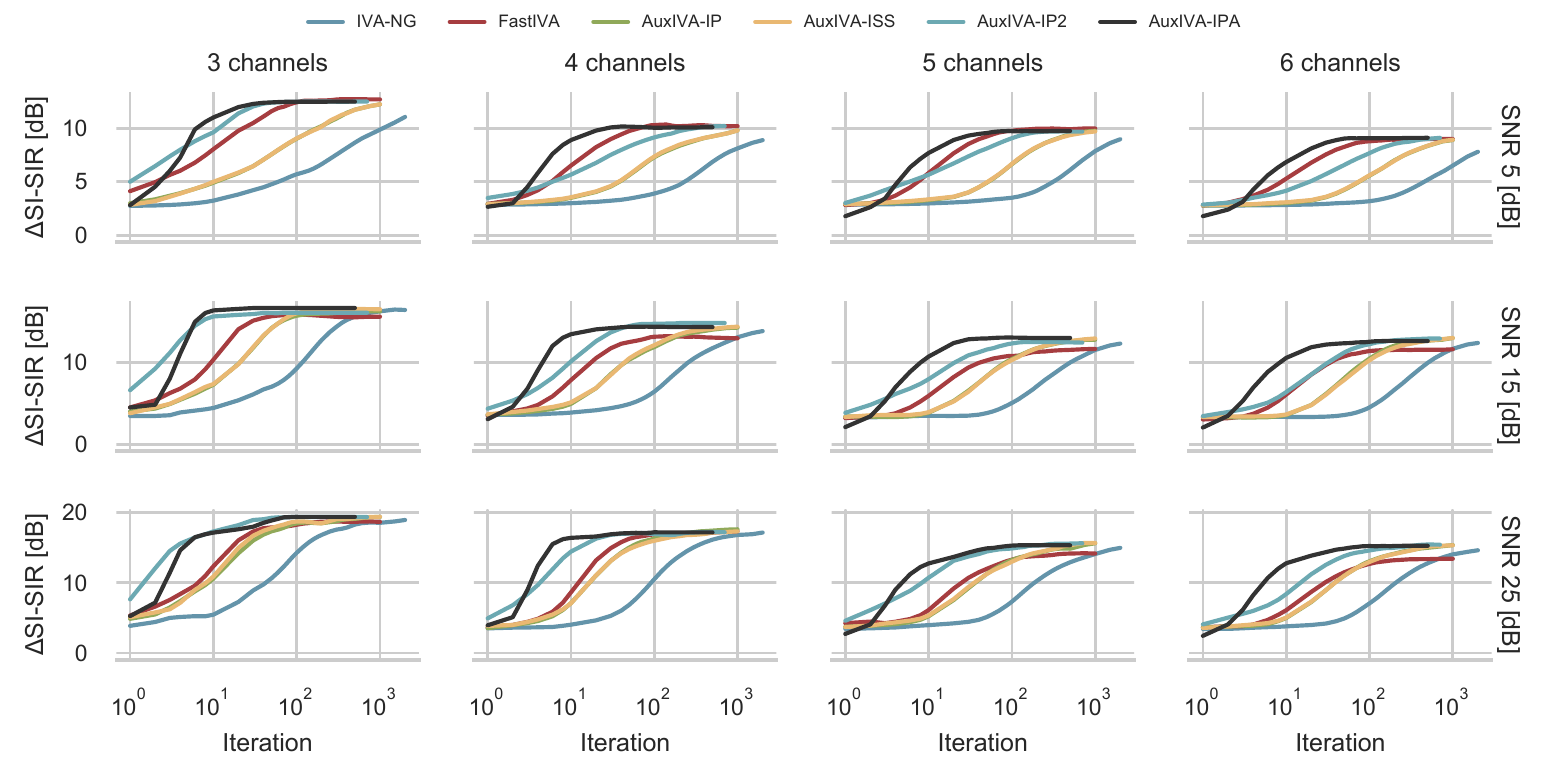}
  \includegraphics{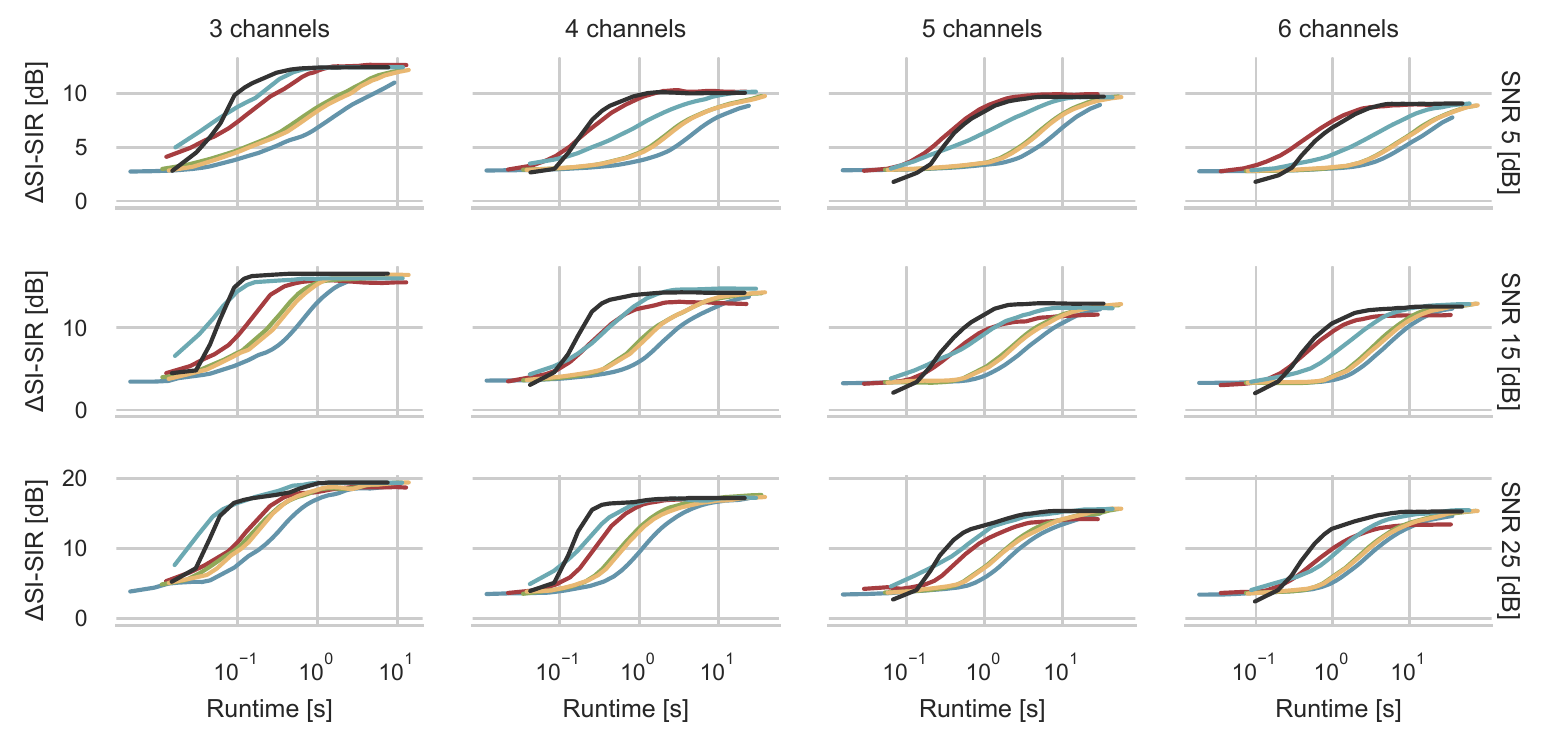}
  \caption{Evolution of the average SI-SIR over number of iterations or runtime in the top and bottom row, respectively. The number of sources increases from three to six from left to right. Note that the lines for AuxIVA-ISS and AuxIVA-IP nearly overlap.}
  \flabel{sep_convergence}
\end{figure*}

We simulate 1000 random rectangular rooms with the \texttt{pyroomacoustics} Python package~\cite{Scheibler:2018di}.
The walls are between \SI{6}{\meter} and \SI{10}{\meter} long, and the ceiling from \SI{2.8}{\meter} to \SI{4.5}{\meter} high.
Simulated reverberation times ($T_{60}$) are approximately uniformly sampled between \SI{60}{\milli\second} and \SI{450}{\milli\second}.
Sources and microphone array are placed at random at least \SI{50}{\centi\meter} away from the walls and between \SI{1}{\meter} and \SI{2}{\meter} high.
The array is circular and regular with 3, 4, 5, or 6 microphones, and radius such that neighboring elements are \SI{10}{\centi\meter} apart.
All sources are placed further from the array than the critical distance of the room --- the distance where direct sound and reverberation have equal energy.
It is computed as $d_{\text{crit}} = 0.057\sqrt{V/T_{60}}\,\si{\meter}$, with $V$ the volume of the room~\cite{Kuttruff:2009uq}.
We define the SNR per microphone as
\begin{align}
\text{SNR}_m = \frac{\E[ \| \tilde{x}_m[\ell] - \tilde{b}_m[\ell] \|^2]}{\E[ \|\tilde{b}_{m}[\ell]\|^2]},
\end{align}
where $\tilde{x}_m[\ell]$ is defined in~\eref{td_recording}.
Uncorrelated Gaussian noise, i.e. $\tilde{b}_m[\ell]$ in \eref{td_recording}, is added to the microphone inputs to obtain a specified SNR at an arbitrary reference microphone.
In all the experiments, we choose the first microphone as the reference, i.e., we fix $\text{SNR}_1$.
We investigate SNR values of \SI{5}{\decibel}, \SI{15}{\decibel}, and \SI{25}{\decibel}.
The simulation is conducted at 16~kHz with concatenated utterances from the CMU Arctic corpus~\cite{Kominek:2004vf,cmu_conc_15}.
We use an STFT with a 4096-points Hamming analysis window and $\nicefrac{3}{4}$-overlap.

The number of iterations of the algorithms are as follows: NG 2000, FastIVA 1000, AuxIVA-IP 1000, AuxIVA-ISS 1000, AuxIVA-IP2 700, AuxIVA-IPA 500.
The demixing matrices are initialized by PCA as in the previous experiment.
The scale of the output is restored by minimizing distortion with respect to the first microphone~\cite{Matsuoka:2001da,Matsuoka:2002da}.
All the experiments are run on a workstation with an Intel\textregistered~ Xeon\textregistered~ Silver 4210 CPU with 40 cores clocked at \SI{2.20}{\giga\hertz}, and \SI{754}{\giga\byte} of RAM.
The algorithms are implemented in Python using Numpy~\cite{Harris:2020fx}, and run in single-threaded environment.

\subsubsection{Results}


First, we compare the final values of the SI-SDR and SI-SIR for all the algorithms.
\ffref{sep_final_value} shows box-plots for different numbers of sources and SNR.
We see that all methods minimizing~\eref{cost_iva} have very similar distributions, indicating similar performance.
Overall at SNR \SI{5}{\decibel}, and for 6 channels at \SI{15}{\decibel}, NG, AuxIVA-IP, and AuxIVA-ISS are not fully converged after the maximum number of iterations, explaining the slightly lower performance.
FastIVA behaves somewhat differently, which may be due to the orthogonality constraint it imposes on the demixing matrix.
At \SI{15}{\decibel} and \SI{25}{\decibel}, it has in nearly all cases lower median SI-SDR and SI-SIR.
However, it performs better than the other algorithms at \SI{5}{\decibel}, where the constraint may help against the noise.
Overall, the SI-SIR is positive in all but some cases (e.g. for 6 channels at \SI{5}{\decibel}), which indicates successful separation.
However, the separated speech quality, as indicated by the SI-SDR, degrades together with the SNR, which is expected.

\begin{table}
  \centering
  \caption{Median runtime in seconds until the cost reduction in one iteration is less than $10M$. We ignored cases where the algorithm was stopped before this occurred.}
  \setlength{\tabcolsep}{3.9pt}  
  \begin{tabular}{@{}lrrrrrrrrrr@{}}
    \toprule
    $M$ & FastIVA &  IP &  ISS &  IP2 &  IPA &    FastIVA & IP &  ISS &  IP2 & IPA  \\
    \midrule
        & \multicolumn{5}{c}{\SI{5}{\decibel}} & \multicolumn{5}{c}{\SI{15}{\decibel}} \\
        \cmidrule(l{2pt}r{2pt}){2-6}           \cmidrule(l{2pt}r{2pt}){7-11}
    3    &1.0 &  8.9 & 12.8 &  0.7 &  0.5  &   0.2 & 3.2  &  3.9 &  0.3 & 0.3  \\
    4    & 3.7 & 34.2 & 37.7 & 26.5 &  2.5 &   3.1 & 19.6 & 21.9 &  7.0 & 1.3  \\
    5    & 4.9 & 52.2 & 56.4 & 40.1 &  7.4 &   7.2 & 46.5 & 51.5 & 15.5 & 3.2  \\
    6    & 7.0 & 73.7 & 77.6 & 60.4 & 15.5 &   9.4 & 69.0 & 73.0 & 43.0 & 7.1  \\
         & \multicolumn{5}{c}{\SI{25}{\decibel}} & & & & & \\
         \cmidrule(l{2pt}r{2pt}){2-6}
    3    & 0.3  &  2.3 &  3.2 &  0.3 & 0.4  & & & & & \\
    4    & 1.3  & 13.0 & 14.1 &  5.2 & 1.2  & & & & & \\
    5    & 6.1  & 35.2 & 38.6 & 12.2 & 2.7  & & & & & \\
    6    & 10.6 & 53.8 & 57.0 & 31.2 & 5.6  & & & & & \\
    \bottomrule
  \end{tabular}
  \tlabel{runtime}
\end{table}

Next, we look at the evolution of the $\Delta$SI-SIR as a function of the number of iterations and runtime in \ffref{sep_convergence}.
This is where AuxIVA-IPA really shines as it outperforms all other methods in nearly all conditions.
As a function of the number of iterations, AuxIVA-IPA is always the fastest.
When measuring the execution time instead, it is the fastest everywhere, except for 5 and 6~channels at \SI{5}{\decibel} where FastIVA has a slight edge.
For 3 channels, there is very little difference between AuxIVA-IP2 and AuxIVA-IPA.
AuxIVA-IPA compares most favorably to other algorithms for 4 and more channels at \SI{15}{\decibel} and \SI{25}{\decibel}.
The dependency of FastIVA on the number of channels seems better, however, it reaches a lower SI-SIR.
\tref{runtime} shows the median runtime needed for the AuxIVA-based algorithms until convergence, defined here as when the decrease of~\eref{cost_iva} in a single step is less than $10M$.
For 4 channels and more, AuxIVA-IPA is between 2.5 to 8.5 times faster than IP2, and 5 to 20 times faster than IP.
We note that the runtime results are limited by the Python/NumPy~\cite{Harris:2020fx} based implementation, and that more efficient implementations may be possible\footnote{For example, AuxIVA-ISS is slower than predicted by its computational complexity. We have tracked this issue to the lower efficiency of the NumPy implementation of the primitives used by ISS, whereas IP/IP2/IPA mostly rely on the highly optimized BLAS primitives.}.

\section{Conclusion}
\seclabel{conclusion}

We proposed a new algorithm for the MM-based independent vector analysis algorithm AuxIVA.
Unlike previous methods that only update part of the demixing matrix at a time, we introduced iterative projection with adjustment (IPA) that updates the whole demixing matrix.
In the derivation of the IPA update, a generic optimization problem, that we call log-quadratically penalized quadratic minimization (LQPQM), appears.
Despite being non-convex, we derived an efficient algorithm to find its global minimum.
To the best of our knowledge, this problem had not been solved before.

We first controlled the performance of the proposed method to minimize the surrogate function of AuxIVA.
We showed that in the first iteration it decreases the surrogate function cost by \SI{11}{\percent} more than the closest other method.
In numerical experiments, we thoroughly investigated the performance of AuxIVA using the different update rules for the separation of synthetic mixtures, and realistically simulated speech mixtures.
In our experiment on synthetic data mixtures, we find that the proposed IPA significantly outperforms other AuxIVA based methods for speed of convergence.
We also find it has the highest success rate of all methods, on par with IP2.
To the best of our knowledge, this is the first time that such an analysis comparing all variants of AuxIVA, FastIVA, and the natural gradient, was performed.
For the practical application of separating speech mixtures, AuxIVA-IPA is the fastest method at mid and high SNR, with no degradation of the separation quality, as measured by standard audio BSS metrics, when compared to other methods.
At low SNR with more channels, AuxIVA-IPA and FastIVA perform similarly.

In future work, we hope to evaluate the impact of IPA updates on more source models, e.g. in ILRMA~\cite{Kitamura:2016vj}, and in the overdetermined~\cite{Scheibler:2019vx,Scheibler:2020mm} and underdetermined~\cite{Sekiguchi:2019by} regimes.
Another interesting question is whether LQPQM is applicable in other contexts.
The log-penalty suggests it might be useful for barrier-based interior point methods.
Another possibility is the maximization of the information theoretic capacity subject to a quadratic penalty or constraint~\cite{Cover:2006ub}.

\section*{Acknowledgment}

I deeply thank the four anonymous reviewers for the time they invested to provide meticulous feedback and pertinent comments.
They saved this paper in more than one way!

I am indebted to Nobutaka Ono for introducing me to AuxIVA in the first place, and pointing me towards the correct way of generating spherically symmetric Laplacean vectors.

Finally, I would like to acknowledge the work of the open source scientific Python community, on which the code for this paper relies.
In particular \texttt{NumPy} for the computations~\cite{Oliphant:2007dm,vanderWalt:2011dp,Harris:2020fx}, \texttt{pandas} for the statistical analysis of the results~\cite{McKinney:2010}, and \texttt{matplotlib} and \texttt{seaborn} for the figures~\cite{Hunter:2007,Waskom:2020}.

\appendices

\section{Determinant of $\mT_k$}
\seclabel{det_Tk}

The proof uses the matrix determinant lemma, and the fact that $\ve_k^\top \bar{\mE}_k \vq = 0$ several times,
\begin{align}
  \det(\mT_k) & = \det(\mI + \ve_k (\vu - \ve_k)^\H + \bar{\mE}_k \vq ^* \ve_k^\top) \nonumber \\
              & = \det\left(\mI_2 + \begin{bmatrix} \vu^\H - \ve_k^\top \\ \ve_k^\top \end{bmatrix} \begin{bmatrix} \ve_k & \bar{\mE}_k \vq ^* \end{bmatrix}\right) \nonumber \\
  & = \det\left( \begin{bmatrix} u_k^* & \vu^\H \bar{\mE}_k \vq ^* \\ 1 & 1 \end{bmatrix}\right)
  =  \vu^\H(\ve_k - \bar{\mE}_k \vq ^*).
  \nonumber
\end{align}

\section{Quadratic form}
\seclabel{quad_form}
Let $\vone$ be the all one vector,
  $a_m = \ve_k \mV_m \ve_k$,
  $b_m = \ve_m \mV_m \ve_k$, and
  $c_m = \ve_m^\top \mV_m \ve_m$.
Further let $\mA = \diag(\ldots,\,a_m,\,\ldots)$, $m\neq k$.
Then,
\begin{align}
  \sum_{m\neq k} & (\ve_m + q_m \ve_k)^\H \mV_m (\ve_m + q_m \ve_k) \nonumber \\
    & = \sum_{m \neq k} a_m |q_m|^2 + (b_m^*q_m + b_m q_m^*) + c_m \nonumber \\
    & = \vq^\H \mA \vq + (\vb^\H \vq + \vq^\H \vb) + \vone^\top \vc \nonumber \\
    & = (\vq + \mA^{-1} \vb)^\H \mA (\vq + \mA^{-1} \vb) - \vb^\H \mA^{-1} \vb + \vone^\top \vc. \nonumber
\end{align}

\section{LQPQM is Bounded from Below}
\seclabel{bounded_below}

\begin{proposition}
  \label{prop:bounded_below}
  The objective function of \eref{lqpqm_eq} is bounded from below and takes its minimum at a finite value.
\end{proposition}
\begin{proof}
We can lower bound the objective in \eref{lqpqm_eq} as follows
\begin{multline}
  \vy^\H\vy - \log \left(\vy^\H\mU\vy + 2 \Re\left\{\vy^\H\mU \vv\right\} + \vv^\H \mU\vv + z\right) \\
        \geq \|\vy\|^2 - \log(a \|\vy\|^2 + b \|\vy\| + c),
\end{multline}
where $a = \lambda_{\max}(\mU)$ is the largest eigenvalue of $\mU$, $b = 2\| \mU \vv\|$, and $c = \vv^\H \mU\vv + z$.
We used the spectral norm of $\mU$ to bound the quadratic term, and Cauchy-Schwarz for the linear term.
Thus, we can equivalently study the real function $f(x) = x^2 - \log(ax^2 + bx + c)$, of $x \geq 0$, with $a>0$, $b, c \geq 0$.
One can show that the stationary points of this function are the zeros of a third order polynomial.
Thus, by the properties of cubic polynomials, $f(x)$ has either one or three stationary points.
Furthermore $f(x) \to +\infty$, when $x\to +\infty$, since the quadratic term grows faster than the log decreases.
Thus, with a single stationary point, $f(x)$ is strictly decreasing to a minimum, and then increasing.
With three stationary points, it must be strictly decreasing, increasing, decreasing, and increasing, with two minima and one maximum.
By continuity, in both cases, $f$ is bounded from below.
\end{proof}

\section{Stationary Points of the LQPQM}
\seclabel{lqpqm_stationary_points}

\begin{lemma}
  \label{lemma:lqpqm_obj_val}
The objective value of~\eref{lqpqm_eq} can be expressed as the function $g(\lambda)$ defined in~\eref{obj_g_lambda}.
\end{lemma}
\begin{proof}
  First, we expand the left-most factor of the second equation in \eref{opt_cond} to obtain,
  \begin{align}
    \lambda & = \vy^\H \mU (\vy + \vv) + \vv^H \mU (\vy + \vv) + z.
    \elabel{opt_cond_2}
  \end{align}
  From the first equation in \eref{opt_cond}, we have
  \begin{align}
    \vy^\H \mU (\vy + \vv) & = \lambda \vy^\H \vy.
    \elabel{match_term_1}
  \end{align}
  Then, by \eref{opt_y}, we find the second term
  \begin{align}
    \vv^\H \mU (\vy + \vv) & = \vv^\H (\mU (\lambda \mI - \mU)^{-1} \mU + \mU) \vv.
    \elabel{match_term_2}
  \end{align}
  Substituting the matching terms of~\eref{opt_cond_2} by~\eref{match_term_1} and~\eref{match_term_2} gives
  \begin{align}
    \lambda & = \lambda \vy^\H \vy + \vv^\H (\mU (\lambda \mI - \mU)^{-1} \mU + \mU) \vv + z.
    \elabel{intermediate_1}
  \end{align}
  Using the eigendecomposition of $\mU$ and rearranging~\eref{intermediate_1},
  \begin{equation}
    \vy^\H \vy = 1 - \sum_{m\in\calS} \frac{\varphi_m|\tilde{v}_m|^2}{(\lambda - \varphi_m)} - \frac{z}{\lambda}.
  \end{equation}
  Finally, replacing into the objective, we obtain~\eref{obj_g_lambda}.
\end{proof}

In the following, to lighten the notation, we assume, without loss of generality, that $\calS = \{1,\ldots,d\}$.

\begin{lemma}
  \label{lemma:f_zeros}
  The function $f(\lambda)$ has
  \begin{enumerate}
    \item no roots smaller or equal to $z$,
    \item zero, one, or two roots in $(z, \varphi_k)$, with $\varphi_k$ being the smallest eigenvalue larger than $z$, if such a root exists,
    \item zero, one, or two roots in $(\varphi_{L-1}, \varphi_{L})$ for $L=k+1, \ldots, d$,
    \item a unique root in the interval $(\max(\varphi_{\max}, z), +\infty)$.
  \end{enumerate}
\end{lemma}
\begin{proof}
  The proof proceeds by inspection of the first and second derivatives of $f(\lambda)$,
  \begin{align}
    f^\prime(\lambda) & = -2\lambda\sum_{m\in\calS} \frac{\varphi_m^2 |\tilde{v}_m|^2}{(\lambda - \varphi_m)^3} - 1, \\
    f^{\prime\prime}(\lambda) & = 2 \sum_{m\in\calS} \varphi_m^2 |\tilde{v}_m|^2 \frac{2 \lambda + \varphi_m}{(\lambda - \varphi_m)^4}.
    \elabel{f_prime_lambda}
  \end{align}
  \begin{enumerate}
    \item Follows from $z - \lambda \geq 0$ in $(0, z)$, and
      \begin{equation}
        \lambda^ 2 \sum_{m\in\calS} \frac{\varphi_m |\tilde{v}_m|^2}{(\lambda - \varphi_m)^2} > 0,\quad \text{if $\lambda > 0$}.
      \end{equation}
        Recall that $\varphi_m \geq 0$, since $\mU$ is positive semi-definite.
    \item In $(z, \varphi_k)$, we have
      \begin{align}
        f(z) > 0, \quad f(\varphi_k - \epsilon) \underset{\epsilon \to 0}{\longrightarrow} +\infty,
      \end{align}
      and because $f^{\prime\prime}(\lambda) > 0$ in this interval, the function there is strictly convex with a unique minimum.
      If the minimum is larger than zero, there is no root.
      If the minimum is zero, there is one root.
      If the minimum is less than zero, there are two roots.
    \item In $(\varphi_{L-1}, \varphi_L)$, we have
      \begin{align}
        f(\varphi_{L-1} + \epsilon) \underset{\epsilon \to 0}{\longrightarrow} +\infty,\quad
        f(\varphi_L - \epsilon) \underset{\epsilon \to 0}{\longrightarrow} +\infty,
      \end{align}
      and $f^{\prime\prime}(\lambda) > 0$, thus, $f(\lambda)$ is strictly convex with a unique minimum, as in 2.
    \item In $(\max(\varphi_{\max}, z), +\infty)$, $f^\prime(\lambda) < 0$ because $\varphi_m > 0$ for all $m$, and $\lambda > \max(\varphi_{\max}, z)$.
      In addition, we have
      \begin{align}
        f(\varphi_{\max} + \epsilon) \underset{\epsilon \to 0}{\longrightarrow} +\infty, \quad\text{and}\quad
        f(\lambda) \underset{\lambda \to +\infty}{\longrightarrow} -\infty,
        \nonumber
      \end{align}
      and thus there is exactly one root in this interval. By 1), the root is in $(z, +\infty)$ if $z > \varphi_{\max}$.
  \end{enumerate}
\end{proof}

\begin{corollary}
  The roots of $f(\lambda)$ are strictly larger than 0.
\end{corollary}
\begin{proof}
  By Lemma~\ref{lemma:f_zeros}, 1), if $f(\lambda) = 0$, then $\lambda > z \geq 0$.
\end{proof}

\begin{fact}
  \label{fact}
  The derivative of $g(\lambda)$ is $g^\prime(\lambda) = \frac{1}{\lambda^2}f(\lambda)$.
\end{fact}

\begin{lemma}
  \label{lemma:f_zeros_between_ev}
  If $f(\lambda)$ has roots $0 < \lambda_1 \leq \lambda_2$ in $(\varphi_{L-1}, \varphi_L)$, then, $g(\lambda_1) \geq g(\lambda_2)$.
\end{lemma}
\begin{proof}
  From Fact~\ref{fact}, we know that the roots of $f(\lambda)$ are stationary points of $g(\lambda)$.
  Moreover, because $f(\lambda)$ is convex with a unique minimum in the interval, $f(\lambda) < 0$ for $\lambda \in (\lambda_1, \lambda_2)$.
  Thus, $g^\prime(\lambda) = \frac{1}{\lambda^2} f(\lambda) < 0$ for $\lambda \in (\lambda_1, \lambda_2)$, and the proof follows.
\end{proof}

\begin{lemma}
  \label{lemma:f_zeros_across_ev}
  Let $\lambda_1 \in (\varphi_{L-1}, \varphi_L)$ and $\lambda_2 \in (\varphi_{L + K}, \varphi_{L+K+1})$ such that $f(\lambda_1) = f(\lambda_2) = 0$, for some $L\in \{1,\ldots,d\}$ and $K\in\{0,\ldots,d-L\} $.
  For convenience, we defined $\varphi_0 = z$ and $\varphi_{d+1} = +\infty$. Then $g(\lambda_1) \geq g(\lambda_2)$.
\end{lemma}
\begin{proof}
  First, we define two functions $\bar{f}_\calA(\lambda)$ and $\bar{g}_\calA(\lambda)$, that are similar to $f(\lambda)$ and $g(\lambda)$, respectively, but with all the discontinuous terms between $\lambda_1$ and $\lambda_2$ removed.
  Then, we show that $\bar{g}_\calA(\lambda)$ is decreasing in $(\lambda_1,\lambda_2)$ with $g(\lambda_1)$ and $g(\lambda_2)$ strictly above and below $\bar{g}_\calA(\lambda)$, respectively.

  Let $\calA = \{L, \ldots, L+K\}$ and define
  \begin{align}
    f_{\calA}(\lambda) & = \lambda^2 \sum_{m \in \calA} \frac{\varphi_m|\tilde{v}_m|^2}{(\lambda - \varphi_m)^2} \geq 0 \\
    g_{\calA}(\lambda) & = - \sum_{m \in \calA} \frac{\varphi_m|\tilde{v}_m|^2}{(\lambda - \varphi_m)}\
    \begin{cases} > 0 & \text{if $\lambda < \varphi_L$} \\ < 0 & \text{if $\lambda > \varphi_{L+K}$} \\ \end{cases}
  \end{align}
  Then, let $\bar{f}_\calA(\lambda) = f(\lambda) - f_\calA(\lambda)$, and $\bar{g}_\calA(\lambda) = g(\lambda) - g_\calA(\lambda)$.
  Note that these two functions are continuous in $(\lambda_1, \lambda_2)$.
  Since $f_\calA(\lambda) \geq 0$, we have
  \begin{align}
    \bar{f}_\calA(\lambda_p) \leq f(\lambda_p) = 0,\quad \text{for $p=1,2$.}
  \end{align}
  Together with Lemma~\ref{lemma:f_zeros}, this means that $\bar{f}_\calA(\lambda)$ has two roots in $(\varphi_{L-1}, \varphi_{L+K+1})$, or just one if $\varphi_{L+K+1} = +\infty$.
  As a consequence, $\bar{g}^\prime_\calA(\lambda) = \frac{1}{\lambda^2} \bar{f}_\calA(\lambda) < 0$ for $\lambda \in (\lambda_1, \lambda_2)$.
  And, thus, $\bar{g}_\calA(\lambda)$ is strictly decreasing on this interval.

  Then, because $g_\calA(\lambda_1) > 0$ and $g_\calA(\lambda_2) < 0$, we have
  \begin{align}
    g(\lambda_1) > \bar{g}_\calA(\lambda_1),\quad \text{and,} \quad g(\lambda_2) < \bar{g}_\calA(\lambda_1),
  \end{align}
  respectively. Finally, because $\bar{g}_\calA(\lambda)$ is strictly decreasing in the interval,
  \begin{equation}
    g(\lambda_1) > \bar{g}_\calA(\lambda_1) > \bar{g}_\calA(\lambda_2) > g(\lambda_2),
  \end{equation}
  which concludes the proof.
\end{proof}

\section{Initialization of the Root Finding Procedure}
\seclabel{root_finding_init}

We propose here a simple scheme providing a good initialization point for the root finding procedure.
Because the inverse square terms in $f(\lambda)$ decay quickly, when $\lambda > \varphi_{\max}$, we can approximate
\begin{equation}
  f(\lambda) \approx \lambda^2 \frac{\varphi_{\max} |v_{\max}|^2}{(\lambda - \varphi_{\max})^2} - \lambda + z
\end{equation}
where $\varphi_d$ is the largest eigenvalue.
Note that this approximation is guaranteed to have its largest zero in the same interval as $f(\lambda)$, which is important for Newton-Raphson.
Equating to zero and multiplying by $(\lambda - \varphi_{\max})^2$ on both sides leads to a cubic equation in $\lambda$ (see also \ffref{secular}),
\begin{multline}
  - \lambda^3 + (\varphi_{\max}|\tilde{v}_{\max}|^2 + 2\varphi_{\max} + z) \lambda^2 \\
  - (\varphi_{\max} + 2z) \varphi_{\max} \lambda + \varphi_{\max}^2 z = 0.
  \elabel{cubic_poly}
\end{multline}
Cubic equations have three solutions including at least one real, and two possibly complex.
We will thus use the largest real solution as a starting point for the root finding.

\bibliographystyle{IEEEtran}
\bibliography{refs}

\begin{thebibliography}{10}
\providecommand{\url}[1]{#1}
\csname url@samestyle\endcsname
\providecommand{\newblock}{\relax}
\providecommand{\bibinfo}[2]{#2}
\providecommand{\BIBentrySTDinterwordspacing}{\spaceskip=0pt\relax}
\providecommand{\BIBentryALTinterwordstretchfactor}{4}
\providecommand{\BIBentryALTinterwordspacing}{\spaceskip=\fontdimen2\font plus
\BIBentryALTinterwordstretchfactor\fontdimen3\font minus
  \fontdimen4\font\relax}
\providecommand{\BIBforeignlanguage}[2]{{%
\expandafter\ifx\csname l@#1\endcsname\relax
\typeout{** WARNING: IEEEtran.bst: No hyphenation pattern has been}%
\typeout{** loaded for the language `#1'. Using the pattern for}%
\typeout{** the default language instead.}%
\else
\language=\csname l@#1\endcsname
\fi
#2}}
\providecommand{\BIBdecl}{\relax}
\BIBdecl

\bibitem{Comon:1512057}
P.~Comon and C.~Jutten, \emph{Handbook of blind source separation: independent
  component analysis and applications}.\hskip 1em plus 0.5em minus 0.4em\relax
  Oxford, UK: Academic Press/Elsevier, 2010.

\bibitem{Makino:2018iq}
S.~Makino, Ed., \emph{Audio source separation}, ser. Signals and Communication
  Technology.\hskip 1em plus 0.5em minus 0.4em\relax Cham, CH: Springer
  International Publishing, 2018.

\bibitem{Makino:2007vg}
S.~Makino, H.~Sawada, and T.-W. Lee, Eds., \emph{Blind Speech Separation}, ser.
  Signals and Communication Technology.\hskip 1em plus 0.5em minus 0.4em\relax
  Cham, CH: Springer, 2007.

\bibitem{Cano:2019dw}
E.~Cano, D.~FitzGerald, A.~Liutkus, M.~D. Plumbley, and F.-R. St\"{o}ter,
  ``Musical source separation: An introduction,'' \emph{IEEE Signal Process.
  Mag.}, vol.~36, no.~1, pp. 31--40, Jan. 2019.

\bibitem{ZARZOSO:1997dt}
V.~Zarzoso, A.~K. Nandi, and E.~Bacharakis, ``Maternal and foetal {ECG}
  separation using blind source separation methods,'' \emph{IMA J Math Appl Med
  Biol}, vol.~14, no.~3, pp. 207--225, Sep. 1997.

\bibitem{Cong2019}
F.~Cong, ``Blind source separation,'' in \emph{EEG Signal Processing and
  Feature Extraction}, L.~Hu and Z.~Zhang, Eds.\hskip 1em plus 0.5em minus
  0.4em\relax Singapore: Springer, 2019, ch.~7, pp. 117--140.

\bibitem{Yang:2018kr}
H.~Yang, H.~Zhang, J.~Li, L.~Yang, and W.~Ding, ``Baseband communication signal
  blind separation algorithm based on complex nonparametric probability density
  estimation,'' \emph{IEEE Access}, vol.~6, pp. 22\,434--22\,440, Apr. 2018.

\bibitem{Comon:1994kr}
P.~Comon, ``Independent component analysis, a new concept?'' \emph{Signal
  Processing}, vol.~36, no.~3, pp. 287--314, 1994.

\bibitem{Hiroe:2006ib}
A.~Hiroe, ``Solution of permutation problem in frequency domain {ICA}, using
  multivariate probability density functions,'' in \emph{Advances in Cryptology
  {\textendash} ASIACRYPT 2016}.\hskip 1em plus 0.5em minus 0.4em\relax Berlin,
  Heidelberg: Springer Berlin Heidelberg, 2006, pp. 601--608.

\bibitem{Kim:2006ex}
T.~Kim, H.~T. Attias, S.-Y. Lee, and T.-W. Lee, ``Blind source separation
  exploiting higher-order frequency dependencies,'' \emph{IEEE Trans. Audio,
  Speech, Lang. Process.}, vol.~15, no.~1, pp. 70--79, Dec. 2006.

\bibitem{Lee:2007bw}
I.~Lee, T.~Kim, and T.-W. Lee, ``Independent vector analysis for convolutive
  blind speech separation,'' in \emph{Blind Speech Separation}.\hskip 1em plus
  0.5em minus 0.4em\relax Dordrecht: Springer, Dordrecht, 2007, pp. 169--192.

\bibitem{Smaragdis:1998kl}
P.~Smaragdis, ``Blind separation of convolved mixtures in the frequency
  domain,'' \emph{Neurocomputing}, vol.~22, no. 1-3, pp. 21--34, Nov. 1998.

\bibitem{Lee:2008dc}
J.-H. Lee, T.-W. Lee, F.~A. Jolesz, and S.-S. Yoo, ``Independent vector
  analysis ({IVA}): Multivariate approach for {fMRI} group study,''
  \emph{NeuroImage}, vol.~40, no.~1, pp. 86--109, Mar. 2008.

\bibitem{Ono:2011tn}
N.~Ono, ``Stable and fast update rules for independent vector analysis based on
  auxiliary function technique,'' in \emph{Proc. IEEE WASPAA}, New Paltz, NY,
  USA, Oct. 2011, pp. 189--192.

\bibitem{Lange:2016wp}
K.~Lange, \emph{{MM} optimization algorithms}.\hskip 1em plus 0.5em minus
  0.4em\relax SIAM, 2016.

\bibitem{Yeredor:hr}
A.~Yeredor, ``On hybrid exact-approximate joint diagonalization,'' in
  \emph{Proc. IEEE CAMSAP}, Dec. 2009, pp. 312--315.

\bibitem{Yeredor:2010kl}
------, ``Blind separation of gaussian sources with general covariance
  structures: Bounds and optimal estimation,'' \emph{IEEE Trans. Signal
  Process.}, vol.~58, no.~10, pp. 5057--5068, Oct. 2010.

\bibitem{Weiss:2017il}
A.~Weiss, A.~Yeredor, S.~Cheema, and M.~Haardt, ``The extended
  {\textquotedblleft}sequentially drilled{\textquotedblright} joint congruence
  transformation and its application in {G}aussian independent vector
  analysis,'' \emph{IEEE Trans. Signal Process.}, vol.~65, no.~23, pp.
  6332--6344, Dec. 2017.

\bibitem{Degerine:2004dy}
S.~Degerine and A.~Zaidi, ``Separation of an instantaneous mixture of
  {G}aussian autoregressive sources by the exact maximum likelihood approach,''
  \emph{IEEE Trans. Signal Process.}, vol.~52, no.~6, pp. 1499--1512, Jun.
  2004.

\bibitem{Lee:2007ct}
I.~Lee, T.~Kim, and T.-W. Lee, ``Fast fixed-point independent vector analysis
  algorithms for convolutive blind source separation,'' \emph{Signal
  Processing}, vol.~87, no.~8, pp. 1859--1871, Aug. 2007.

\bibitem{Yatabe:2018vb}
K.~Yatabe and D.~Kitamura, ``Determined blind source separation via proximal
  splitting algorithm,'' in \emph{Proc. IEEE ICASSP}, Calgary, CA, Apr. 2018,
  pp. 776--780.

\bibitem{Yatabe:2020tv}
------, ``Determined {BSS} based on time-frequency masking and its application
  to harmonic vector analysis,'' \emph{IEEE/ACM Trans. Audio Speech Lang.
  Process.}, Apr. 2021, early access.

\bibitem{Gu:2019hl}
Z.~Gu, J.~Lu, and K.~Chen, ``Speech separation using independent vector
  analysis with an amplitude variable {G}aussian mixture model,'' in
  \emph{Proc. Interspeech 2019}, Graz, AU, Sep. 2019, pp. 1358--1362.

\bibitem{Kitamura:2016vj}
D.~Kitamura, N.~Ono, H.~Sawada, H.~Kameoka, and H.~Saruwatari, ``Determined
  blind source separation unifying independent vector analysis and nonnegative
  matrix factorization,'' \emph{IEEE/ACM Trans. Audio Speech Lang. Process.},
  vol.~24, no.~9, pp. 1626--1641, Sep. 2016.

\bibitem{Kameoka:2019be}
H.~Kameoka, L.~Li, S.~Inoue, and S.~Makino, ``Supervised determined source
  separation with multichannel variational autoencoder,'' \emph{Neural
  computation}, vol.~31, no.~9, pp. 1891--1914, Sep. 2019.

\bibitem{Makishima:2019fl}
N.~Makishima, S.~Mogami, N.~Takamune, D.~Kitamura, H.~Sumino, S.~Takamichi,
  H.~Saruwatari, and N.~Ono, ``Independent deeply learned matrix analysis for
  determined audio source separation,'' \emph{IEEE/ACM Trans. Audio Speech
  Lang. Process.}, vol.~27, no.~10, pp. 1601--1615, 2019.

\bibitem{Shin:2020fg}
U.-H. Shin and H.-M. Park, ``Auxiliary-function-based independent vector
  analysis using generalized inter-clique dependence source models with clique
  variance estimation,'' \emph{IEEE Access}, vol.~8, pp. 68\,103--68\,113, Apr.
  2020.

\bibitem{Scheibler:2019vx}
R.~Scheibler and N.~Ono, ``Independent vector analysis with more microphones
  than sources,'' in \emph{Proc. IEEE WASPAA}, New Paltz, NY, USA, Oct. 2019,
  pp. 185--189.

\bibitem{Ono:2010hh}
N.~Ono and S.~Miyabe, ``Auxiliary-function-based independent component analysis
  for super-{G}aussian sources,'' \emph{Proc. LVA/ICA}, vol. 6365, no.~6, pp.
  165--172, Sep. 2010.

\bibitem{Ono:2012wa}
N.~Ono, ``Fast stereo independent vector analysis and its implementation on
  mobile phone,'' in \emph{Proc. IWAENC}, Aachen, DE, Sep. 2012.

\bibitem{Scheibler:2019tt}
R.~Scheibler and N.~Ono, ``Fast independent vector extraction by iterative
  {SINR} maximization,'' in \emph{Proc. IEEE ICASSP}, Barcelona, ES, May 2020,
  accepted.

\bibitem{Ikeshita:2020ic}
R.~Ikeshita, T.~Nakatani, and S.~Araki, ``Overdetermined independent vector
  analysis,'' in \emph{Proc. IEEE ICASSP}, Barcelona, ES, May 2020, accepted.

\bibitem{Ono:2018wk}
N.~Ono, ``Fast algorithm for independent component/vector/low-rank matrix
  analysis with three or more sources,'' in \emph{Proc. Acoustical Society of
  Japan}, Mar. 2018, pp. 437--438.

\bibitem{Scheibler:2020mm}
R.~Scheibler and N.~Ono, ``{MM} algorithms for joint independent subspace
  analysis with application to blind single and multi-source extraction,''
  \emph{arXiv}, Apr. 2020, arXiv:2004.03926.

\bibitem{Scheibler:2020ig}
------, ``Fast and stable blind source separation with rank-1 updates,'' in
  \emph{Proc. IEEE ICASSP}, Barcelona, ES, May 2020, pp. 236--240.

\bibitem{Golub:1973ha}
G.~H. Golub, ``Some modified matrix eigenvalue problems,'' \emph{SIAM Review},
  vol.~15, no.~2, pp. 318--334, Apr. 1973.

\bibitem{Bunch:1978bn}
J.~R. Bunch, C.~P. Nielsen, and D.~C. Sorensen, ``Rank-one modification of the
  symmetric eigenproblem,'' \emph{Numerische Mathematik}, vol.~31, no.~1, pp.
  31--48, Mar. 1978.

\bibitem{Yu:1991cy}
K.-B. Yu, ``Recursive updating the eigenvalue decomposition of a covariance
  matrix,'' \emph{IEEE Trans. Signal Process.}, vol.~39, no.~5, pp. 1136--1145,
  May 1991.

\bibitem{MORE:1993dy}
J.~J. More, ``Generalizations of the trust region problem,'' \emph{Optim.
  Method Softw.}, vol.~2, no. 3-4, pp. 189--209, Jan. 1993.

\bibitem{Lorenz:em}
R.~G. Lorenz and S.~P. Boyd, ``Robust minimum variance beamforming,''
  \emph{IEEE Trans. Signal Process.}, vol.~53, no.~5, pp. 1684--1696, 2005.

\bibitem{Beck:2012dj}
A.~Beck, P.~Stoica, and J.~Li, ``Exact and approximate solutions of source
  localization problems,'' \emph{IEEE Trans. Signal Process.}, vol.~56, no.~5,
  pp. 1770--1778, Apr. 2008.

\bibitem{Togami:2020doa}
M.~Togami and R.~Scheibler, ``Sparseness-aware {DOA} estimation with
  majorization minimization,'' in \emph{Proc. Interspeech}, Shanghai, CN, Oct.
  2020, pp. 5046--5050.

\bibitem{Anderson:2014jd}
M.~Anderson, G.-S. Fu, R.~Phlypo, and T.~Adal{\i}, ``Independent vector
  analysis: Identification conditions and performance bounds,'' \emph{IEEE
  Trans. Signal Process.}, vol.~62, no.~17, pp. 4399--4410, Aug. 2014.

\bibitem{Wu:1983hp}
C.~F.~J. Wu, ``On the convergence properties of the {EM} algorithm,'' \emph{The
  Annals of Statistics}, vol.~11, no.~1, pp. 95--103, Mar. 1983.

\bibitem{DeLeeuw:1997vw}
J.~de~Leeuw and W.~J. Heiser, ``Convergence of correction matrix algorithms for
  multidimensional scaling,'' in \emph{Geometric Representations of Relational
  Data}, J.~C. Lingoes, E.~Roskam, and I.~Borg, Eds.\hskip 1em plus 0.5em minus
  0.4em\relax Ann Arbor, MI: Mathesis Press, 1977, pp. 735--752.

\bibitem{Daubechies:2010hf}
I.~Daubechies, R.~DeVore, M.~Fornasier, and C.~{Sinan G{\"u}nt{\"u}rk},
  ``Iteratively reweighted least squares minimization for sparse recovery,''
  \emph{Communications on Pure and Applied Mathematics}, vol.~63, no.~1, pp.
  1--38, Jan. 2010.

\bibitem{Yamaoka:2019wb}
K.~Yamaoka, R.~Scheibler, N.~Ono, and Y.~Wakabayashi, ``Sub-sample time delay
  estimation via auxiliary-function-based iterative updates,'' in \emph{Proc.
  IEEE WASPAA}, New Paltz, NY, USA, Oct. 2019, pp. 130--134.

\bibitem{Hunter:2004hq}
D.~R. Hunter and K.~Lange, ``A tutorial on {MM} algorithms,'' \emph{The
  American Statistician}, vol.~58, no.~1, pp. 30--37, Feb. 2004.

\bibitem{Sun:2017hb}
Y.~Sun, P.~Babu, and D.~P. Palomar, ``Majorization-minimization algorithms in
  signal processing, communications, and machine learning,'' \emph{IEEE Trans.
  Signal Process.}, vol.~65, no.~3, pp. 794--816, Feb. 2017.

\bibitem{Yeredor:2012gv}
A.~Yeredor, B.~Song, F.~Roemer, and M.~Haardt, ``A
  {\textquotedblleft}sequentially drilled{\textquotedblright} joint congruence
  ({SeDJoCo}) transformation with applications in blind source separation and
  multiuser {MIMO} systems,'' \emph{Signal Processing, IEEE Transactions on},
  vol.~60, no.~6, pp. 2744--2757, May 2012.

\bibitem{Murata:2001gb}
N.~Murata, S.~Ikeda, and A.~Ziehe, ``An approach to blind source separation
  based on temporal structure of speech signals,'' \emph{Neurocomputing},
  vol.~41, no. 1-4, pp. 1--24, Oct. 2001.

\bibitem{Matsuoka:2002da}
K.~Matsuoka, ``Minimal distortion principle for blind source separation,'' in
  \emph{Proc. SICE}, Aug. 2002, pp. 2138--2143.

\bibitem{Remmert:1991vp}
R.~Remmert, \emph{Theory of complex functions}.\hskip 1em plus 0.5em minus
  0.4em\relax New York: Springer Science+Business Media, 1991.

\bibitem{Hyvarinen:ek}
A.~Hyv\"{a}rinen, ``Fast and robust fixed-point algorithms for independent
  component analysis,'' \emph{IEEE Trans. Neural Netw.}, vol.~10, no.~3, pp.
  626--634, May 1999.

\bibitem{LeRoux:2018tq}
J.~Le~Roux, S.~Wisdom, H.~Erdogan, and J.~R. Hershey, ``{SDR} --- half-baked or
  well done?'' in \emph{Proc. IEEE ICASSP}, Brighton, UK, May 2019, pp.
  626--630.

\bibitem{Allen:1977in}
J.~Allen, ``Short term spectral analysis, synthesis, and modification by
  discrete {F}ourier transform,'' \emph{IEEE Trans. Acoust., Speech, Signal
  Process.}, vol.~25, no.~3, pp. 235--238, Jun. 1977.

\bibitem{Griffin:1984cq}
D.~Griffin and J.~Lim, ``Signal estimation from modified short-time {F}ourier
  transform,'' \emph{IEEE Trans. Acoust. Speech Signal Process.}, vol.~32,
  no.~2, pp. 236--243, 1984.

\bibitem{Scheibler:2018di}
R.~Scheibler, E.~Bezzam, and I.~Dokmani{\'c}, ``Pyroomacoustics: A {P}ython
  package for audio room simulations and array processing algorithms,'' in
  \emph{Proc. IEEE ICASSP}, Calgary, CA, Apr. 2018, pp. 351--355.

\bibitem{Kuttruff:2009uq}
H.~Kuttruff, \emph{Room acoustics}.\hskip 1em plus 0.5em minus 0.4em\relax CRC
  Press, 2009.

\bibitem{Kominek:2004vf}
J.~Kominek and A.~W. Black, ``{CMU ARCTIC} databases for speech synthesis,''
  Language Technologies Institute, School of Computer Science, Carnegie Mellon
  University, Tech. Rep. CMU-LTI-03-177, 2003.

\bibitem{cmu_conc_15}
\BIBentryALTinterwordspacing
R.~Scheibler, ``{CMU ARCTIC} concatenated 15s,'' {Z}enodo. [Online]. Available:
  \url{http://doi.org/10.5281/zenodo.3066489}
\BIBentrySTDinterwordspacing

\bibitem{Matsuoka:2001da}
K.~Matsuoka and S.~Nakashima, ``Minimal distortion principle for blind source
  separation,'' in \emph{Proc. ICA}, San Diego, Dec. 2001, pp. 722--727.

\bibitem{Harris:2020fx}
C.~R. Harris, K.~J. Millman, S.~J. van~der Walt, R.~Gommers, P.~Virtanen,
  D.~Cournapeau, E.~Wieser, J.~Taylor, S.~Berg, N.~J. Smith, R.~Kern, M.~Picus,
  S.~Hoyer, M.~H. van Kerkwijk, M.~Brett, A.~Haldane, J.~F. del R{\'\i}o,
  M.~Wiebe, P.~Peterson, P.~G{\'e}rard-Marchant, K.~Sheppard, T.~Reddy,
  W.~Weckesser, H.~Abbasi, C.~Gohlke, and T.~E. Oliphant, ``Array programming
  with {NumPy},'' \emph{Nature}, vol. 585, no. 7825, pp. 357--362, Sep. 2020.

\bibitem{Sekiguchi:2019by}
K.~Sekiguchi, A.~A. Nugraha, Y.~Bando, and K.~Yoshii, ``Fast multichannel
  source separation based on jointly diagonalizable spatial covariance
  matrices,'' \emph{Proc. EUSIPCO}, Sep. 2019.

\bibitem{Cover:2006ub}
T.~M. Cover and J.~A. Thomas, \emph{Elements of Information Theory}.\hskip 1em
  plus 0.5em minus 0.4em\relax Hoboken, NJ, USA: John Wiley \& Sons, Jul. 2006.

\bibitem{Oliphant:2007dm}
T.~E. Oliphant, ``{Python} for scientific computing,'' \emph{Computing in
  Science {\&} Engineering}, vol.~9, no.~3, pp. 10--20, 2007.

\bibitem{vanderWalt:2011dp}
S.~van~der Walt, S.~C. Colbert, and G.~Varoquaux, ``The {N}um{P}y array: A
  structure for efficient numerical computation,'' \emph{Computing in Science
  {\&} Engineering}, vol.~13, no.~2, pp. 22--30, Feb. 2011.

\bibitem{McKinney:2010}
{W}es {M}c{K}inney, ``{D}ata structures for statistical computing in python,''
  in \emph{Proc. 9th {P}ython Sci. Conf.}, {S}t\'efan van~der {W}alt and
  {J}arrod {M}illman, Eds., 2010, pp. 56 -- 61.

\bibitem{Hunter:2007}
J.~D. Hunter, ``Matplotlib: A 2{D} graphics environment,'' \emph{Computing in
  Science \& Engineering}, vol.~9, no.~3, pp. 90--95, 2007.

\bibitem{Waskom:2020}
\BIBentryALTinterwordspacing
M.~Waskom, O.~Botvinnik, J.~Ostblom, M.~Gelbart, S.~Lukauskas, P.~Hobson, D.~C.
  Gemperline, T.~Augspurger, Y.~Halchenko, J.~B. Cole, and et~al.,
  ``{mwaskom/seaborn: v0.10.1 (April 2020)},'' Apr 2020. [Online]. Available:
  \url{https://github.com/mwaskom/seaborn}
\BIBentrySTDinterwordspacing

\end{thebibliography}

\begin{IEEEbiography}[{\includegraphics[width=1in,height=1.25in,clip,keepaspectratio]{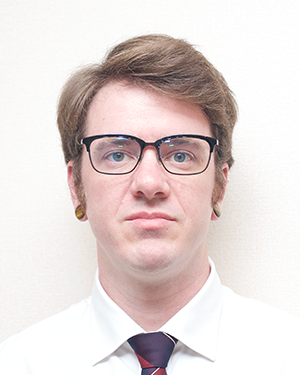}}]{Robin Scheibler} (M'07, SM'20) is a senior researcher at LINE Corporation.
  Robin received his B.Sc, M.Sc, and Ph.D. from Ecole Polytechnique Fédérale de Lausanne (EPFL, Switzerland).
  He also worked at the research labs of NEC Corporation (Kawasaki, Japan) and IBM Research (Z\"{u}rich, Switzerland).
  From 2017 to 2019, he was a post-doctoral fellow at the Tokyo Metropolitan University, and then a specially appointed associate professor until February 2020.
  Robin's research interests are in efficient algorithms for signal processing, and array signal processing more particularly.
  He also likes to build large microphone arrays and is the lead developer of pyroomacoustics, an open source library for room acoustics simulation and array signal processing.
\end{IEEEbiography}

\end{document}